\documentclass[10pt,journal]{IEEEtran}

\ifCLASSOPTIONcompsoc

  \usepackage[nocompress]{cite}
\else
  \usepackage{cite}
\fi

\ifCLASSINFOpdf
\else
\fi

\usepackage{calc}

\usepackage{color}
\usepackage{graphicx}
\usepackage{amsmath}
\usepackage{amssymb}
\usepackage{verbatim}
\usepackage{array}
\usepackage{textcomp}
\usepackage{slashbox}
\usepackage{pifont}
\usepackage[ruled,vlined]{algorithm2e}

\usepackage{amsmath}
\usepackage{amsthm}
\usepackage{amssymb}
\usepackage{multicol}
\usepackage{multirow}
\usepackage{graphicx}
\usepackage[mathscr]{euscript}
\usepackage{multirow}
\usepackage{centernot}
\usepackage{listliketab}
\usepackage{cutwin}
\usepackage{verbatim}
\usepackage{slashbox}
\usepackage{color}
 \usepackage{multirow}
 \usepackage{wrapfig}
 \usepackage{bm}

\usepackage{array}

\makeatletter
\renewcommand{\maketag@@@}[1]{\hbox{\m@th\normalsize\normalfont#1}}%
\makeatother

\newtheorem{theorem}{\bf Theorem}
\newtheorem{lemma}{\bf Lemma}

\newtheorem{proposition}{\bf Proposition}
\newtheorem{remark}{\bf Remark}

\newtheorem{definition}{\bf Definition}

\newcommand\T{\rule{0pt}{2.1ex}}
\newcommand\B{\rule[-0.7ex]{0pt}{0pt}}

\DeclareMathOperator*{\argmax}{arg\,max}
\DeclareMathOperator*{\argmin}{arg\,min}
\DeclareMathOperator*{\avg}{avg\,}

\newtheorem{innercustomthm}{Theorem}
\newenvironment{customthm}[1]
  {\renewcommand\theinnercustomthm{#1}\innercustomthm}
  {\endinnercustomthm}
\newtheorem{innercustomprop}{Proposition}
\newenvironment{customprop}[1]
  {\renewcommand\theinnercustomprop{#1}\innercustomprop}
  {\endinnercustomprop}

\begin{document}
\title{Modeling Spread of Preferences in Social Networks for Sampling-based Preference Aggregation}
\author{
Swapnil Dhamal, Rohith D. Vallam, and Y. Narahari
\thanks{

The original version of this paper is accepted for publication in IEEE Transactions on Network Science and Engineering. The copyright for this article belongs to IEEE.
DOI 10.1109/TNSE.2017.2772878

Contact author: Swapnil Dhamal (swapnil.dhamal@gmail.com).
S.  Dhamal and R. D. Vallam were with the Department
of Computer Science and Automation, Indian Institute of Science, Bangalore, when most of the work was done.
A part of the work was done when S. Dhamal was with T\'el\'ecom SudParis, CNRS, Universit\'e Paris-Saclay, France, and R. D. Vallam was with IBM India Research Labs, Bangalore.
Y. Narahari is with the Department
of Computer Science and Automation, Indian Institute of Science, Bangalore.
}
}

\IEEEtitleabstractindextext{%
\begin{abstract}
Given a large population, it is an intensive task to gather individual preferences over a set of alternatives and arrive at an aggregate or collective preference of the population. We show that social network underlying the population can be harnessed to accomplish this task effectively, by sampling preferences of a small subset of representative nodes. We first develop a Facebook app to create a dataset consisting of preferences of nodes and the underlying social network, using which, we develop models that capture how preferences are distributed among nodes in a typical social network. We hence propose an appropriate objective function for the problem of selecting best representative nodes. We devise two algorithms, namely, Greedy-min which provides a performance guarantee for a wide class of popular voting rules, and Greedy-sum which exhibits excellent performance in practice. We compare the performance of these proposed algorithms against random-polling and popular centrality measures, and provide a detailed analysis of the obtained results. Our analysis suggests that selecting representatives using social network information is advantageous for aggregating preferences related to personal topics (e.g., lifestyle), while random polling with a reasonable sample size is good enough for aggregating preferences related to social topics (e.g., government policies).
\end{abstract}

}

\maketitle

\IEEEdisplaynontitleabstractindextext

\IEEEpeerreviewmaketitle

{\section{Introduction}}
\label{sec:intro_pasn}

{T}{here} are several scenarios such as elections, opinion polls, public project initiatives, funding decisions, etc., where a population (society) faces a number of alternatives. In such scenarios, the population's collective preference over the given alternatives is of importance.
Ideally, one would want to obtain the preferences of all the individuals in the population and aggregate them so as to represent the population's preference. These individuals can be termed as `voters' in this context. This process of computing an aggregate preference over a set of alternatives, given individual preferences, is termed {\em preference aggregation} (a well-studied topic in social choice theory).
It is generally assumed that the preferences of all the voters are known.
In real-world scenarios, however, it may not be feasible to gather the individual preferences of all the voters owing to  factors such as their lack of interest to provide prompt, truthful, well-informed, and well-thought out preferences over the given alternatives. 
As an immediate example,
consider a company desiring to launch a future product based on its customer feedback. However, very few customers might be willing to 
devote the needed effort to promptly provide a useful feedback.
Further, even if we manage to gather preferences of all the voters, obtaining an aggregate preference may not be feasible owing to computational issues.

 Owing to the difficulty involved in obtaining individual preferences of an entire population and in computing the aggregate preference,
 an attractive approach would be to select a subset of voters 
 whose preferences reflect the population's preferences,
 and hence incentivize only those voters to report their preferences.
 We refer to these voters as {\em representatives}.
 To determine such representatives, this work proposes ideas for harnessing additional information regarding the population of voters: the underlying social network.
 It has been established based on a large number of empirical evidences over the years that, there is a significant correlation between the preferences of  voters and the underlying social network; this can be attributed to the {\em homophily} phenomenon \cite{networkscrowdsmarkets}. 
 There have been several efforts for modeling homophily and strength of ties in social networks \cite{mcpherson2001birds,xiang2010modeling}.
 The social network and the involved tie strengths (similarities among nodes) could thus give additional information regarding the individual preferences.
 
 In the context of social networks, we use `nodes' to represent voters and `neighbors' to represent their connections.
 
 \vspace{3mm}
 \subsection{Overview of the Problem and Solution Approach}
 \label{sec:overview}
 
 We now informally describe the problem addressed in this paper.
 Given a population and a topic (such as political party) for which alternatives need to be ranked, let $p_A$ be the collective (or aggregate) preference of the population computed using a certain voting (or aggregation) rule. Note that unless we have the preferences of the nodes in the population, $p_A$ is indeterminate. Suppose instead of obtaining preferences of all the nodes, we obtain preferences of only a subset of nodes and compute their collective preference, say $p_B$. The problem we study is to find a subset of a certain cardinality $k$ such that $p_B$ is as close as possible to $p_A$. For determining how close $p_B$ is to $p_A$, it is required that we deduce $p_A$ (since it is unknown to us). To do this, we need an underlying model for deducing the preferences of nodes, with which we can compute an estimate of $p_A$ using the voting rule. Motivated by the correlation that exists between the preferences of nodes and the underlying social network, we propose  a model that captures how preferences are distributed in the network.
 Such a model would enable us to deduce the individual preferences of all the nodes, with the knowledge of preferences of a subset of nodes.
 With such a model in place, we would be able to estimate how close $p_B$ is to $p_A$, and also find a subset of nodes for which this `closeness' is optimal.

 The problem of determining the best representatives can thus be classified into two subproblems: (a) deducing how preferences are distributed in the population and (b) finding a subset of nodes whose collective preference closely resembles or approximates the collective preference of the population. This paper addresses these two subproblems.

 \subsection{Preliminaries}
 \label{sec:prelim_pasn}
 
 We now provide some preliminaries on preference aggregation, required throughout the paper.
 Given a set of alternatives, 
 we refer to a ranked list of alternatives 
 as a {\em preference} and the multiset consisting of the individual preferences as {\em preference profile}.
 For example, if the set of alternatives is $\{X,Y,Z\}$ and node $i$ prefers $Y$ the most and $X$ the least, $i$'s preference 
 is written as
 $(Y,Z,X)_i$.
 Suppose node $j$'s preference is $(X,Y,Z)_j$, then the preference profile of the population $\{i,j\}$ is $\{(Y,Z,X),(X,Y,Z)\}$.
 A widely used measure of dissimilarity between two preferences is {\em Kendall-Tau distance} which counts the number of pairwise inversions with respect to the alternatives.
 Given $r$ alternatives,
 the normalized Kendall-Tau distance can be obtained by dividing this distance by \begin{scriptsize}$\dbinom{r}{2}$\end{scriptsize}, the maximum distance between  any two preferences on $r$ alternatives.
 For example, the Kendall-Tau distance between preferences $(X,Y,Z)$ and $(Y,Z,X)$ is 2, since two pairs, $\{X,Y\}$ and $\{X,Z\}$, are inverted between them. The normalized Kendall-Tau distance is $2/$\begin{scriptsize}$\dbinom{3}{2}$\end{scriptsize}. %

 An {\em aggregation rule} takes a preference profile as input and outputs the {\em aggregate preference},
 which aims to reflect the collective opinion of all the nodes. We do not assume any tie-breaking rule in order to avoid bias towards any particular alternative while determining the aggregate preference.
 So an aggregation rule may output multiple aggregate preferences (that is, aggregation rule is a correspondence).
 A survey of voting rules
 and their extensions to
  aggregation rules,
 can be found in \cite{brandt2012computational}.
 We consider several aggregation rules for our study, 
 namely, Bucklin, Smith set, Borda, Veto, Minmax (pairwise opposition), Dictatorship, Random Dictatorship, Schulze, Plurality, Kemeny, and Copeland.

 \subsection{Relevant Work}
 \label{sec:relevant_pasn}
 
 This work is at the interface of social network analysis and voting theory. Specifically, we adapt a model of opinion dynamics in social networks, so as to model how preferences are distributed in a network. Using the adapted model, we sample a subset of nodes which could act as representatives of the network.
 The research topics of relevance to our work are thus, relation between social networks and voting, models of opinion dynamics, and network sampling.
 We now present the relevant literature and position our work.

 \vspace{3mm}
 \subsubsection{Social networks and voting}
 The pioneering Columbia and Michigan political voting research 
 \cite{sheingold1973social} emphasizes on the importance of the underlying social network. 
 It has been observed that the social network has higher impact on one's political party choice than background attributes like class or ethnicity~\cite{burstein1976social}
 as well as
 religion, education, and social status~\cite{Nieuwbeerta2000313}. 
 It has been argued that 
 interactions in social networks have a strong, though often overlooked, influence on voting, since interactions allow an individual to gather information beyond personal resource constraints
 \cite{McClurg2006electoral,zuckerman2005social}.
 The impact of social networks has also been compared with that of mass media communication with respect to voting choices, where it is observed that social discussions outweigh the media effect \cite{beck2002calculus}, and that both the effects should be studied together \cite{campus2008social}.
 On the other hand, it has also been argued via a maximum likelihood approach to political voting, that it is optimal to ignore the network structure~\cite{conitzer2012should}.
 It is also not uncommon to observe highly accurate predictions about election outcomes by soliciting opinions of randomly chosen voters.

 There are 
 supporting arguments for both views (regarding whether or not social network plays a role in voting); so one of the goals of this work is to 
 identify conditions under which social network visibly plays a role.
 As we will see,
 our work suggests that social network and homophily indeed play a critical role in voting related to personal topics (such as lifestyle). However,
 for voting related to social topics (such as government policies), 
 factors which are external to the network and common to all nodes (such as mass media channels) could play a strong role,
 which would allow ignoring social network if the sample size is reasonable.

  \vspace{3mm}
 \subsubsection{Models of opinion dynamics}
 \label{sec:relevant_diffusion}

 Opinion dynamics is the process of development of opinions in a population over time, primarily owing to the information exchanged through interactions.
 Several models of opinion dynamics in networks have been studied in the literature
 \cite{acemoglu2011opinion,lorenz2007continuous,networkscrowdsmarkets}, such as
  DeGroot,
  Friedkin-Johnsen,
 Independent Cascade, Linear Threshold, etc.
 Since opinions and networks are key factors in opinion dynamics as well as our setting, it is natural to consider the applicability of opinion dynamics techniques to our setting.
 
 One of the primary factors that distinguishes our setting from opinion dynamics
 is that, the edge weights in opinion dynamics represent influence weights or probabilities, while the edge weights in our setting represent probability distributions over similarity values (homophilic tie strengths).
 So if one intends to use the models of opinion dynamics for the problem under consideration, 
 they need to be appropriately adapted.
 In this paper, we propose an adaptation of the Independent Cascade model; instead of deducing the probability of a node getting influenced (given edge probabilities and a subset of nodes which are already influenced), we deduce the preference of a node (given homophilic tie strengths and the preferences of a subset of nodes).

  \vspace{3mm}
 \subsubsection{{Network sampling}}
 \label{sec:relevant_sampling}

 A taxonomy of different graph
 sampling objectives and approaches is provided in \cite{hu2013survey}.
 Problems in
 developing a general theory of network sampling are discussed in \cite{granovetter1976network}.
 A discussion 
  on which sampling method to use and how small the sample size can be, is provided in \cite{leskovec2006sampling}; there it is also experimentally observed that simple uniform random
 node selection performs surprisingly well.
 The problem of finding a subset of users to statistically represent the
 original social network has also been studied \cite{tang2015sampling,sun2013learning}.

 Our setting principally deviates from the existing literature in that, as stated earlier, we consider
 edges having probability distributions over similarity values (instead of weights or probabilities). 
 Moreover, we address the problem of sampling best representative nodes based on 
 a ranked list of alternatives, 
 and
 undertake a study to determine the performance of sampling-based approach for preference aggregation with respect to a number of voting rules.

 Furthermore, the problem of node selection based on ranked list of alternatives has been studied by utilizing the attributes of voters and alternatives~\cite{grum2013uai}, however, the underlying network is ignored.

 \subsection{Our Contributions}
 \label{sec:contrib_pasn}
 
 Our specific contributions are as follows.

 \begin{itemize}
 [\setlength{\labelwidth}{\widthof{\textbullet}}
 \setlength{\labelsep}{7.5pt}
  \setlength{\IEEElabelindent}{0pt}
     \IEEEiedlabeljustifyl ]
     
     \setlength\itemsep{.5em}
     
 \item 
 We develop a Facebook app to create a dataset consisting of the preferences of nodes over a variety of topics and the underlying social network.
 With this dataset, 
 we propose and validate a number of simple yet faithful models with the aim of capturing how preferences are distributed among nodes in a social network,
 while having probabilistic information regarding tie similarities.

 \item 
 We formulate an objective function for the problem of determining best representative nodes in a social network,
 and hence 
 propose a robustness property for aggregation rules
 to justify appropriateness of 
 two surrogate objective functions for computational purpose.

 \item 
 We provide a guarantee on the performance of one of the algorithms (Greedy-min), 
 and study desirable properties of a second algorithm (Greedy-sum) which exhibits excellent performance in practice. 
 We compare these algorithms with the popular random polling method  and  well-known centrality measures over a number of voting rules, and analyze the results.
 
 \item
 We present insights on the effectiveness of our approach 
 with respect to personal and social topics.

 \end{itemize}

 We believe the results in this paper offer a rigorous model for capturing spread of preferences in a social network, and present effective methods for sampling-based preference aggregation using social networks.

  \vspace{2mm}
 \section{Modeling Spread of Preferences in a \\Social Network}
 \label{sec:modeling}
 
 In this section, we first introduce the idea of modeling the spread of preferences in a social network, with an analogy to modeling opinion dynamics. We then describe the dataset created through our Facebook app and hence develop a number of simple and intuitive yet faithful models for deducing the spread of preferences across the nodes in a social network.
 Here we use the term `spread' to indicate distribution, not diffusion (we do not use the term `distribution' of preferences since we frequently use it in the context of probability distributions).
 We treat an edge as representing similarity 
 (tie strength with respect to homophily) 
 between two connected nodes, and not as influence weights or probabilities as in opinion dynamics models. 
 The preferences of the nodes could have resulted owing to
  opinion dynamics over the network, along with external influences such as mass media channels.
 Our goal here is to model the spread of these preferences while knowing the edge weights indicating similarities.

 \subsection{Dataset for Modeling Spread of Preferences in a Social Network}

 \begin{table}[b]
 \caption{Statistics of the Facebook dataset}
 \begin{center}
 \begin{tabular}{|c|c|}
 \hline \T \B 
 number of nodes & 844
 \\ \hline \T \B 
 number of edges & 6129
 \\ \hline \T \B 
 average clustering coefficient & 0.5890
 \\ \hline \T \B 
 number of triangles & 33216
 \\ \hline \T \B 
 fraction of closed triangles & 0.4542
 \\ \hline \T \B 
 diameter & 13 
 \\ \hline \T \B 
  90-percentile effective diameter & 4.9
  \\ \hline \T \B 
  power law exponent & 2.12
 \\ \hline 
 \end{tabular}
 \end{center}
 \label{tab:fbstats}
 \end{table}

 \begin{table}[t]
 \caption{Notation used in the paper
 }
 \begin{center}
 \begin{tabular}{|l|l|}
 \hline 
 \T \B 
 $N$ & set of nodes in the network \\ \hline
 \T \B 
 $E$ & set of edges in the network \\ \hline
 \T \B 
 $n$ & number of nodes in the network \\ \hline
 \T \B 
 $m$ & number of edges in the network \\ \hline
 \T \B 
 $r$ & number of alternatives \\ \hline
 \T \B 
 $P$ & preference profile of $N$ \\ \hline
 \T \B 
 $\tilde{d}(x,y)$ & distance between preferences $x$ and $y$ \\ \hline
 \T \B 
 $f$ & preference aggregation rule \\ \hline
 \T \B 
 $i,j$ & typical nodes in the network \\ \hline
 \T \B 
 $d(i,j)$ & expected distance between preferences of $i$ and $j$ \\ \hline
 \T \B 
 $c(i,j)$ & $1-d(i,j)$ \\ \hline
 \T \B 
 $k$ & desired cardinality of the representative set \\ \hline
 \T \B 
 $M$ & set of representatives who report their preferences \\ \hline
 \T \B
 $\Phi(M,i)$ & representative of node $i$ in set $M$ \\ \hline
 \T \B 
 $Q$ & profile containing unweighted preferences of $M$ \\ \hline
 \T \B 
 $Q'$ & profile containing weighted preferences of $M$ \\ \hline
 \T \B 
 $\Delta$ & error operator between aggregate preferences \\ \hline
 \end{tabular}
 \end{center}
 \label{tab:notation}
 \end{table}

 We presented a connection between the opinion dynamics setting and our setting in
 Section \ref{sec:relevant_diffusion}.
 It is assumed while studying opinion dynamics that the edge influence weights or probabilities are given, using which, a model (such as Independent Cascade and Linear Threshold) predicts the probability with which each node would get influenced starting with a given seed set. Alternatively, models such as DeGroot would predict the final opinion of each node, based on the edge influence weights which are assumed to be given. On similar lines, for our proposed models in Section \ref{sec:model_spread}, we assume the edge similarity parameters to be given, using which we deduce the preference of each node. 
 It can also be viewed that,
 given the edge parameters indicating homophilic tie strengths, we aim to find preferences of all nodes such that they are consistent with the given parameters for all edges, 
  that is, how preferences are distributed or spread across the nodes in the network.
 These edge parameters, however, would need to be inferred from some available data.
 This, again, is on similar lines as opinion dynamics models, where the influence weights or probabilities would be required to be inferred from some past data.

 In order to develop  models for achieving the aforementioned objectives, there was a need of a dataset that consists of (a) preferences of nodes for a range of topics and (b) the underlying social network.
 With this underlying goal, we developed a Facebook app titled {\em The Perfect Representer}, which asked the app users to report their preferences for 8 topics, over 5 alternatives for each topic. 
 The topics (broadly classified into personal and social types), their alternatives, the network of app users, and details about dataset preprocessing, are
 provided in Appendix \ref{app:facebook_app}.
 Table \ref{tab:fbstats} provides some statistics of this dataset.
 Table~\ref{tab:notation} presents the notation used throughout the paper.

 \subsection{Modeling Distance between Preferences of Nodes}
 \label{sec:model_pairs}

 In order to study the 
 homophilic tie strength or similarity between two nodes with respect to their preferences,
 we use the measure of normalized Kendall-Tau distance.
 The histogram of distance between preferences of most pairs of nodes (connected as well as unconnected), over the considered topics, followed a bell curve (most distances were clubbed together, with few of them spread apart). As a preliminary fit to the data, we consider Gaussian distribution since it is a natural and most commonly observed distribution in real-world applications. Since the range of values taken by the distance is bounded between 0 and 1, we consider Truncated Gaussian distribution. Furthermore, as the range of values is discrete, we consider a discrete version of truncated Gaussian distribution; denote it by $\mathcal{D}$.
 The discretization can be done in the following way.
 We know that when the number of alternatives is $r$, 
 the distance between consecutive discrete values is $1/$\begin{scriptsize}$\dbinom{r}{2}$\end{scriptsize} $=\frac{2}{r(r-1)}$.
 Let $F$ be the cumulative distribution function of the continuous truncated Gaussian distribution; the value of the probability mass function of $\mathcal{D}$ at $x$ can be shown to be
 
 \begin{small}
 \begin{displaymath}
 F\left(\min\left\{x+\frac{1}{r(r-1)},1\right\}\right) - F\left(\max\left\{x-\frac{1}{r(r-1)},0\right\}\right).
 \end{displaymath}
 \end{small}

 So for a pair of nodes $\{i,j\}$, there is an associated histogram of normalized Kendall-Tau distances over different topics.
 To this histogram,  we attempt to fit distribution $\mathcal{D}$ using MLE, and infer the corresponding parameters ($\mu_{ij}$ and $\sigma_{ij}$ of the original Gaussian distribution from which $\mathcal{D}$ is derived).
 We use KL divergence
 to find the fitting error between the true histogram obtained from the data and the fitted distribution $\mathcal{D}$.
 This process is run for all pairs of nodes; the resulting histogram of KL divergences over all pairs is presented in Figure \ref{fig:KLdiv_hist}.
 The root mean square (RMS) KL divergence for connected as well as unconnected pairs
 was observed to be 0.25 with more than 90\% pairs having KL divergence less than 0.6.

 Let 
 the expected distance between nodes $i$ and $j$ with respect to their preferences, be denoted by $d(i,j)$.
 Let {\em distance matrix} be a matrix whose cell $(i,j)$ is $d(i,j)$ and {\em similarity matrix} be a matrix whose cell $(i,j)$ is $c(i,j) = 1-d(i,j)$.
 Following are
 certain statistics about $d(i,j)$'s over all pairs of nodes in the obtained Facebook app data.
 The means ($=\avg_{\{i,j\}} d(i,j)$) for personal, social, and both types of topics are respectively 0.40, 0.30, 0.35, while the standard deviations are respectively 0.12, 0.08, 0.09.
 (Our focus throughout this paper will be on aggregating preferences across all topics or issues;
 we provide a preliminary analysis by considering different types of topics separately, in Section~\ref{sec:pvss}).

 \begin{figure}[t]
 \centering
 \includegraphics[scale=.56]{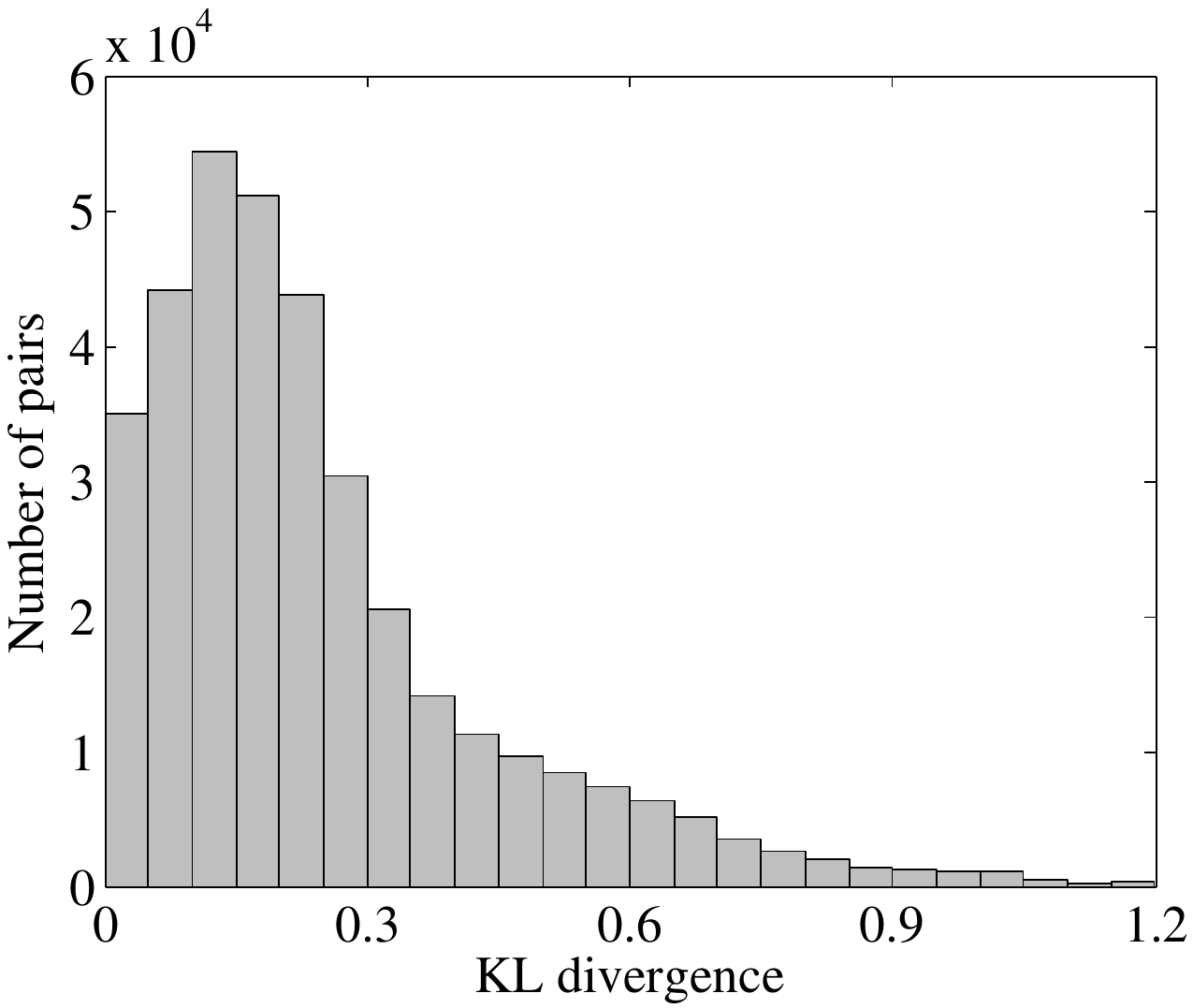}
 \caption{Histogram of KL divergence between the normalized histograms and the distribution $\mathcal{D}$ obtained using MLE for all pairs of nodes}
 \label{fig:KLdiv_hist}
 \end{figure}
 
 \subsection{Modeling Spread of Preferences in Social Network}
 \label{sec:model_spread}
 
 Recollect that the primary objective of modeling the spread of preferences in a social network is to identify best set of representative nodes for the entire social network. In order to do this, we need not only the distribution of distances between connected nodes, but also that between unconnected nodes. 
 As an analogy to modeling opinion dynamics as per models such as Independent Cascade, in order to identify best set of seed nodes for maximizing opinion diffusion, we need the influence of a candidate set, not only on its neighbors but also on distant nodes.

 Given the preferences of a set of nodes (call it {\em initializing set}), our models aim to deduce the possible preferences of all the nodes in the social network. If this model is run for several iterations, say $\mathbb{T}$, with a randomized initializing set in each iteration, we would have deduced the preferences of nodes for these $\mathbb{T}$ generated (or simulated) topics (and hence $\mathbb{T}$ preference profiles). This would then enable us to deduce the distribution of distances between unconnected nodes as well.
 We address each of our models as a {\em Random Preferences Model (RPM)}
 since the deduced preferences and hence the distances are randomized.

 The models that we propose work iteratively (similar to the Independent Cascade model wherein iterations correspond to time steps). In each iteration, we partition the nodes into two sets, namely, (1) {\em assigned nodes} which are assigned a preference, and (2) {\em unassigned nodes} which are not assigned a preference as yet.
 Let {\em potentially next nodes} in a given iteration be the subset of unassigned nodes, which have at least one neighbor in the set of assigned nodes.
 Starting with the nodes in the initializing set as the only assigned nodes,
 a node is chosen uniformly at random from the set of potentially next nodes in each iteration, and is assigned a preference based on the preferences of its {\em assigned neighbors} (neighbors belonging to the set of assigned nodes).
 Algorithm~\ref{alg:generic_model} presents a generic model
  on these lines.
 
 Note that the preference profiles are obtained by considering independently sampled initializing sets across iterations (versus considering the same set in models such as Independent Cascade).
 Also, social networks are known to have a core-periphery structure. Considering a random core node and a random periphery node, the latter is more likely to deduce its preference due to the preference of the former, than the other way around; this is captured over iterations in Algorithm~\ref{alg:generic_model}.
 We now present a number of models as its special cases.
 
 \begin{algorithm}[t]
 \begin{small}
 \KwIn{Connected graph $G$ with parameters $\mu,\sigma$ on its  \mbox{edges,
 Number of generated (simulated) topics $\mathbb{T}$}}
 \KwOut{Preference profiles for $\mathbb{T}$ generated topics}
 \For{$t \gets 1$ \textbf{to} $\mathbb{T}$}{
 \mbox{Randomly choose an initializing set of certain size $s$\;}
 Assign preferences to nodes in this initializing set\;
 \For{$i \gets 1$ \textbf{to} $n-s$}{
 Choose an unassigned node $u$ uniformly at random from the set of potentially next nodes\;
 Assign a preference to $u$ based on:\\
  (i) the model under consideration and \\
  (ii) either (a) preferences of  assigned neighbors \\~~~or (b) preference of one of its assigned  \\~~~neighbors chosen based on a certain criterion\;
 }
 }
 \end{small}
 \caption{{A generic model for spread of preferences in a social network}}
 \label{alg:generic_model}
 \end{algorithm}

 \vspace{3mm}
 \subsubsection{{Independent Conditioning (RPM-IC)}}
 Let $P_j$ be the random preference to be assigned to a node $j$ and $A_j$ be the set of {\em assigned neighbors} of node $j$. So given the preferences of its assigned neighbors, the probability of node $j$ being assigned a preference $p_j$ is 
 
 \begin{small}
 \begin{align*}
 & \mathbb{P}\left(P_j = p_j | (P_{i} = p_{i})_{i \in A_j}\right)
 \\&= \frac{\mathbb{P}\left((P_{i} = p_{i})_{i \in A_j} | P_j = p_j\right) \; \mathbb{P}(P_j = p_j)}{\sum_{p_j} \mathbb{P}\left((P_{i} = p_{i})_{i \in A_j} | P_j = p_j\right) \;  \mathbb{P}(P_j = p_j)}
 \;\;\;\text{(Bayes' rule)}
 \\&\propto \mathbb{P}\left((P_{i} = p_{i})_{i \in A_j} | P_j = p_j\right)
 \end{align*}
 \end{small}

 The proportionality results since the denominator is common, and $\mathbb{P}(P_j = p_j) = \frac{1}{r!}$ for all $p_j$'s (assuming no prior bias).
 Now we make a simplifying assumption of mutual independence among the preferences of assigned neighbors of node $j$, given its own preference. So the above proportionality results in
 %
 
 \begin{small}
 \begin{align}
 \label{eqn:rpmic}
 & \mathbb{P}\left(P_j = p_j | (P_{i} = p_{i})_{i \in A_j}\right)
 \propto \prod_{i \in A_j} \mathbb{P}(P_{i} = p_{i} | P_j = p_j)
 \end{align}
 \end{small}
  \noindent
  The right hand term consists of factors $\mathbb{P}(P_{i} = p_{i} | P_j = p_j)$ for each $i \in A_j$.
 $\mathbb{P}(P_{i} = p_{i} | P_j = p_j)$ says, given we have assigned preference $p_j$ to $j$, what is the probability that $i$ had preference $p_i$?
  We now show how to compute it.
  Let $D_{ij}$ be the random variable corresponding to the distance between  $i$ and $j$ (so $D_{ij}$ has distribution $\mathcal{D}$ with values of $\mu,\sigma$ corresponding to pair $\{i,j\}$). Let $\tilde{d}(p_i,p_j)$ be the distance between preferences $p_i$ and $p_j$. So,
 
 %
  \begin{footnotesize}
 \begin{align}
 \hspace{-1mm}
 \nonumber
 &~\mathbb{P}\left(P_{i} = p_{i} | P_j = p_j\right) 
 \\
 \hspace{-1mm}
 \nonumber
 &= \mathbb{P}\left(P_{i} = p_{i} , D_{ij} = \tilde{d}(p_i,p_j) | P_j = p_j\right) 
 \nonumber
 \;\;
 (\because \text{given } P_j, P_i \cap D_{ij} = P_i)
 \\
 \hspace{-1mm}
 \nonumber
 &= \mathbb{P}\left(D_{ij} = \tilde{d}(p_i,p_j) | P_j = p_j\right) \;\mathbb{P}\,\Big(P_{i} = p_{i} | D_{ij} = \tilde{d}(p_i,p_j) , 
 \nonumber 
 P_j = p_j\Big)
 \\
 \hspace{-1mm}
 \label{eqn:conditional}
 &= \mathbb{P}\left(D_{ij} = \tilde{d}(p_i,p_j)\right) \;\mathbb{P}\left(P_{i} = p_{i} | D_{ij} = \tilde{d}(p_i,p_j) , P_j = p_j\right)
 \\
 \hspace{-1mm}
 \nonumber
 &\;\;\;\;\;\;\;\;\;\;\;\;\;\;\;\;\;\;\;\;\;\;\;\;\;\;\;\;\;\;\;\;\;\;\;\;\;\;\;\;\;\;\;\;\;\;\;\;\;\;\;\;\;\;\;\;\;\;\;\;\;\;(\because D_{ij} \text{ is independent of } P_j)
 \end{align}
 \end{footnotesize}
 
 Here, $\mathbb{P}\left(D_{ij} = \tilde{d}(p_i,p_j)\right)$ can be readily obtained by looking at the distribution $\mathcal{D}$ corresponding to $\{i,j\}$.
 Also, as we assume that no preference has higher priority than any other, $\mathbb{P}\left(P_{i} = p_{i} | D_{ij} = \tilde{d}(p_i,p_j) , P_j = p_j\right)$ is precisely the reciprocal of the number of preferences which are at distance $\tilde{d}(p_i,p_j)$ from a given preference. 
 This value can be expressed in terms of distance $\tilde{d}(p_i,p_j)$ and the number of alternatives.
 As an example for the case of 5 alternatives, the number of preferences which are at a normalized Kendall-Tau distance of $0.1$ (or Kendall-Tau distance of 1) from any given preference is 4; for example, if the given preference is $(A,B,C,D,E)$, the 4 preferences are $(B,A,C,D,E)$, $(A,C,B,D,E)$, $(A,B,D,C,E)$, $(A,B,C,E,D)$. It is clear that this count is independent of the given preference.

 The initializing set for this model is a singleton set chosen uniformly at random, and is assigned a preference chosen uniformly at random from the set of all preferences.
 For each unassigned node, this model computes probabilities for each of the $r!$ possible preferences, by looking at the preferences of its assigned neighbors, and hence chooses exactly one preference based on the computed probabilities (multinomial sampling). So the time complexity of this model for assigning preferences for $\mathbb{T}$ topics is $O(r!(\sum_{i\in N} \text{deg}(i))\mathbb{T}) = O(r!m\mathbb{T})$, where $\text{deg}(i)$ is the degree of node $i$.
 Note that the initializing set is taken to be singleton, since
 an initializing set consisting of multiple nodes may lead to conflict in preferences, and hence inconsistencies with the distributions.
 
 \vspace{3mm}
 \subsubsection{{Sampling (RPM-S)}}

 For assigning preference to an unassigned node $j$, we first choose one of its assigned neighbors $i$; say its assigned preference is $p_i$. Then we sample a value from the discretized truncated Gaussian distribution $\mathcal{D}$ having the parameters $(\mu_{ij},\sigma_{ij})$; say the sampled value is $\hat{d}_{ij}$. Following this, we choose a preference $p_j$ such that the Kendall Tau distance between $p_i$ and $p_j$, that is $\tilde{d}(p_i,p_j)$, is $\hat{d}_{ij}$ (in case of multiple possibilities, choose $p_j$ uniformly at random from among these possibilities).
 The assigned neighbor $i \in A_j$ could be selected in multiple ways; we enlist three natural ways which we consider in our study:

       \vspace{3mm}
 \begin{enumerate}[\setlength{\labelwidth}{\widthof{\textbullet}}
 \setlength{\labelsep}{7.5pt}
  \setlength{\IEEElabelindent}{0pt}
     \IEEEiedlabeljustifyl ]
 \item[a)] \textit{Random:} A node is selected uniformly at random from $A_j$. 
 This is a natural way and is immune to overfitting.
 \item[b)] \textit{$\mu$-based:} A node $i$ is selected from $A_j$ with probability proportional to $1-\mu_{ij}$ (multinomial sampling). 
 This is consistent with the empirical belief that a node's preference depends more on its more similar neighbors.
 \item[c)] \textit{$\sigma$-based:} A node $i$ is selected from $A_j$ with probability proportional to $1/\sigma_{ij}$ (multinomial sampling). 
 This is statistically a good choice because, giving lower priority to distributions with low standard deviations may result in extremely large errors.
 \end{enumerate}
 
 Like RPM-IC, 
 the initializing set for this model also is a singleton set chosen uniformly at random, and is assigned a preference chosen uniformly at random from the set of all preferences.
 For each unassigned node, this model selects an assigned neighbor in one of the above ways, samples a distance value from the corresponding distribution, and chooses a preference uniformly at random from the set of preferences which are at that distance from the preference of the selected assigned neighbor. So the time complexity of this model for assigning preferences for $\mathbb{T}$ topics is $O(r!(\sum_{i\in N} \text{deg}(i))\mathbb{T}) = O(r!m\mathbb{T})$.

      \vspace{3mm}
 \subsubsection{{Duplicating (RPM-D)}}
 In this model, node $j$ is assigned a preference by duplicating the preference of its most similar assigned neighbor. 
 This model pushes the similarity between a node and its most similar assigned neighbor to the extreme extent that, the preference to be assigned to the former is not just similar to the latter, but is exactly its copy due to imitation.

 Here, if the initializing set is a singleton, all nodes would have the same preference (since the graph is assumed to be connected).
 A small-sized initializing set also would lead to very few distinct preferences.
 So as a heuristic,
 the initializing set for this model is a connected set of certain size $s$ which is obtained using the following iterative approach: start with a node chosen uniformly at random and then continue adding a new node to the set uniformly at random, from among the nodes that are connected to at least one node in the set. In our experiments, we choose $s$ itself to be uniformly at random from $\{1,\ldots,\lceil\sqrt{n}\rceil\}$
 (similar to the information diffusion setting where $\lceil\sqrt{n}\rceil$ is usually a good heuristic upper bound for the size of the seed set). 
 The nodes in this initializing set are assigned preferences based on RPM-IC.
 The time complexity of this model for assigning preferences for $\mathbb{T}$ topics is $O((\sum_{i\in N} \text{deg}(i))\mathbb{T}) = O(m\mathbb{T})$.

      \vspace{3mm}
 \subsubsection{{Random (RPM-R)}}
 In this model, preferences are assigned randomly to all the nodes without considering the distribution of distances from their neighbors, that is, without taking the social network effect into account. This model can be refined based on some known bias in preferences, for instance, 
 there may be a prior distribution on preferences owing to common external effects such as mass media.
 Its time complexity for assigning the preferences for $\mathbb{T}$ topics is $O(n\mathbb{T})$.

      \vspace{3mm}
 \subsubsection{{Mean Similarity Model - Shortest Path Based \\(MSM-SP)}}
 \label{sec:msm_sp}
 
 Unlike the models discussed so far, this model does not deduce the spread of preferences in a social network. Instead, it deduces the mean similarity between any pair of nodes, given the mean similarities of connected nodes.
 
 Recall that cell $(i,j)$ of a {\em distance matrix} contains $d(i,j)$, the expected distance between preferences of nodes $i$ and $j$. 
 We initialize all values in this matrix to $0$ for $i=j$ and to $1$ (the upper bound on distance) for any unconnected pair $\{i,j\}$.
 In the case of a connected pair $\{i,j\}$, the value $d(i,j)$ is initialized to the actual observed expected distance (this value is known from the edge parameters indicating homophilic tie strength).
 Following the initialization of the distance matrix, the next step is to update it.

 Consider nodes $(i,v,j)$ where we know the expected distances $d(v,i)$ and $d(v,j)$, and we wish to find $d(i,j)$ via node $v$.
 Given the preference of node $v$ and $d_x=d(v,i)$, let the preference of node $i$ be chosen uniformly at random from the set of preferences that are at a distance $\eta$ from the preference of node $v$, where $\eta$ is drawn from distribution $\mathcal{D}$  with mean $d_x$ (and some standard deviation).
 Similarly, given $d_y=d(v,j)$, let the preference of node $j$ be obtained.
 Using this procedure, the distance between the obtained preferences of nodes $i$ and $j$ via $v$ over several iterations and varying values of standard deviations, was observed to follow a bell curve; so we again approximate this distribution by $\mathcal{D}$.
 Let the corresponding expected distance constitute the cell $(d_x,d_y)$ of a table, say $T_r$, where $r$ is the number of alternatives (for the purpose of forming a table, we consider only finite number of values of $d_x,d_y$).
 It is clear that this distance is independent of the actual preference of node $v$.

 We empirically observe that $T_r$ is different from $T_{r'}$ for $r \neq r'$.
 Following are the general properties of $T_r$:
 \begin{itemize}
 \item $T_r(d_y,d_x) = T_r(d_x,d_y)$
 \item $T_r(1-d_x,d_y) = T_r(d_x,1-d_y) = 1-T_r(d_x,d_y)$
 \item $T_r(1-d_x,1-d_y) = T_r(d_x,d_y)$
 \end{itemize}

 \begin{table}[t]
 \caption{A partial view of table $T_5$ 
 }
 \begin{center}
 \begin{small}
 \begin{tabular}{|>{\small}c|>{\small}c|>{\small}c|>{\small}c|>{\small}c|>{\small}c||>{\small}c|}
 \hline 
 0.00 &  0.10 &  0.20 &  0.30 &  0.40 &  0.50 & \slashbox{$d_x$}{$d_y$}	\\ \hline  \hline
      0.00   &  0.10  &  0.20  &  0.30  &  0.40  &  0.50 &  0.00 \\ \hline
      \multicolumn{1}{c|}{}  &  0.17  &  0.26  &  0.33  &  0.42  &  0.50 &  0.10  \\ \cline{2-7}
      \multicolumn{2}{c|}{}  &  0.32  &  0.37  &  0.43  &  0.50 &  0.20\\ \cline{3-7}
      \multicolumn{3}{c|}{}  &  0.40  &  0.45  &  0.50   &  0.30\\ \cline{4-7}
      \multicolumn{4}{c|}{}  &  0.47  &  0.50   &  0.40\\ \cline{5-7}
      \multicolumn{5}{c|}{}  &  0.50  &  0.50\\ \cline{6-7}
 \end{tabular}
 \end{small}
 \end{center}
 \label{tab:distance_table}
 \end{table}

 As the topics of our app had 5 alternatives, we obtain the table $T_5$ 
 for any pair $\{d_x,d_y\}$.
 In order to consider finite number of values of $d_x,d_y$ for forming the table, we only account for values that are multiples of 0.01 (and also round every entry in $T_r$ to the nearest multiple of 0.01).
 Table~\ref{tab:distance_table} presents a partial view of $T_5$ which can be completed using the general properties of $T_r$ enlisted above; we present $d_x,d_y$ in multiples of 0.10 for brevity.
 Now the next question is to find $d(i,j)$ for any pair $\{i,j\}$.
 In order to provide a fit to the distances obtained from the dataset, 
 we initialize the distance matrix as explained in the beginning of this subsection (while rounding every value to the nearest multiple of 0.01) and update it
 based on the {\em all pairs shortest path algorithm}~\cite{cormen2009introduction} with the following update rule:
 \begin{center}
 {\small 
 \textbf{if} $d(v,i)~\text{\textcircled{+}}_r~d(v,j) < d(i,j)$ \textbf{then} $d(i,j) = d(v,i)~\text{\textcircled{+}}_r~d(v,j)$,}
 \end{center}
 where we define operator $\text{\textcircled{+}}_r$ as follows:
 \begin{small}
 \begin{displaymath}
   d_x~\text{\textcircled{+}}_r~d_y=\begin{cases}
     T_r(d_x,d_y), & \text{if $d_x \leq 0.5$ and $d_y \leq 0.5$}\\
     \max\{d_x,d_y\}, & \text{if $d_x > 0.5$ or $d_y > 0.5$}
   \end{cases}
 \end{displaymath}
 \end{small}
 The corresponding {similarity matrix} is obtained by assigning $1-d(i,j)$ to its cell $(i,j)$.
 The two cases while defining $\text{\textcircled{+}}_r$ 
 ensure that $d(i,j)$ via $v$ is assigned a value which is at least $\max\{d(v,i),d(v,j)\}$.
 This is to guarantee the convergence of the adapted all pairs shortest path algorithm.

 The time complexity of deducing the mean distances between all pairs of nodes using MSM-SP is dominated by the all pairs shortest path algorithm, $O(n^2 \log n + nm)$ using Dijkstra's algorithm over all source nodes, where $m$ is generally small owing to sparsity of social networks.
 \begin{remark}
 We have already seen the time complexities of the models (other than MSM-SP) for assigning preferences for $\mathbb{T}$ topics; the time complexity of deducing the mean distances between all pairs of nodes after that, is $O(\mathcal{G}_rn^2\mathbb{T})$, where $\mathcal{G}_r$ is the time complexity of computing the distance between two preferences (for instance, $\mathcal{G}_r$ is $O(r^2)$ for Kendall-Tau distance, $O(r)$ for Spearman Footrule distance).
 \end{remark}

 \vspace{-4mm}
 \subsection{Validating the Models}
 
 To validate a given model, we generated the preferences of all the nodes for $\mathbb{T}=10^4$ simulated topics. Following this, we could get the distances between preferences of every pair of nodes in terms of normalized Kendall-Tau distance. 
 In order to measure the error $\text{err}(\{i,j\})$ of this deduced model distribution against the distribution $\mathcal{D}$ with the actual values of $\mu_{ij}$ and $\sigma_{ij}$ for a particular pair of nodes $\{i,j\}$, we used the following methods:
 \begin{enumerate}[\setlength{\labelwidth}{\widthof{\textbullet}}
 \setlength{\labelsep}{7.5pt}
  \setlength{\IEEElabelindent}{0pt}
     \IEEEiedlabeljustifyl ]
 \item Kullback-Leibler (KL) divergence, a standard way of measuring error between  actual and model distributions,
 \item Earth Mover's distance (EMD), another measure of the distance between two probability distributions,
 \item Absolute difference between the means of these two distributions, since some of our algorithms would be working with the distribution means,
 \item Chi-Square Goodness of Fit Test, a non-parametric test to compare the obtained distribution with the expected probability distribution (we use significance level 0.05).

 \end{enumerate}

 For methods 1-3,
 we measured the total error over all pairs of nodes as the root mean square (RMS) error, that is, $\sqrt{\avg_{\{i,j\}} [\text{err}(\{i,j\})]^2}$.
 Figure~\ref{fig:plot_models} provides a comparison among the models under study, with respect to their errors and running times. 
 (Note that RMS KL divergence and EMD are not applicable for MSM-SP).
 RPM-IC gave the least errors but at the cost of extremely high running time.
 RPM-D and RPM-R ran fast but their errors were in a higher range.
 RPM-S showed a good balance between the errors and running time; the way of choosing the assigned neighbor ($\mu$-based, $\sigma$-based, or random) did not show significant effect on its results. 
  MSM-SP was observed to be the best model when our objective was to deduce the mean distances between all pairs of nodes, and not the preferences themselves.
 With 5 alternatives, the normalized Kendall-Tau distances can take 11 possible values $\{0,0.1,\ldots,0.9,1\}$. 
 The upper critical value of chi-square distribution with 10 degrees of freedom for significance level 0.05 is 18.307. 
 The test is failed if the chi-square test statistic exceeds this  value.

 \begin{figure}[t!]
 \centering
 \includegraphics[scale=.52]{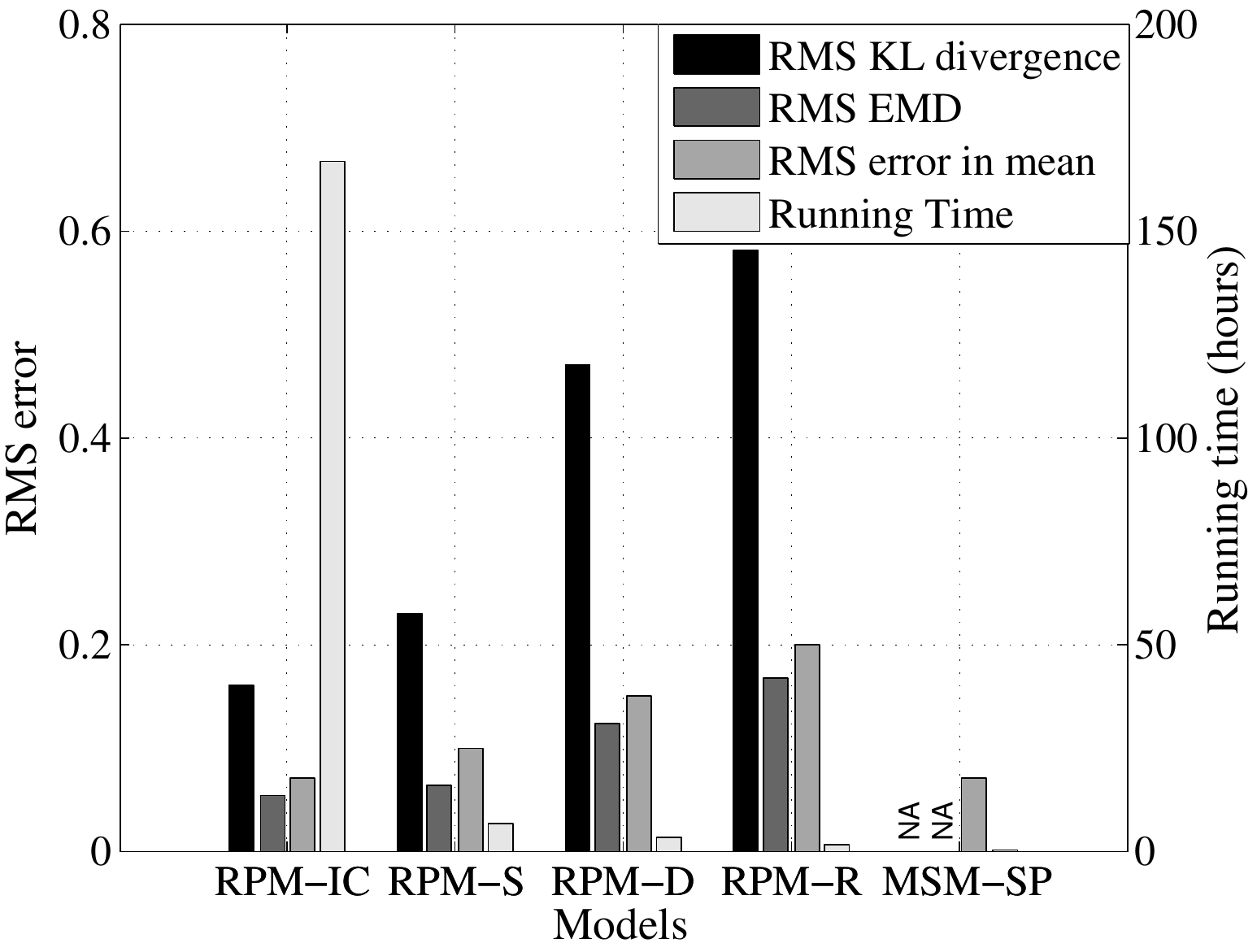}
 \caption{Comparison among the considered models when run for $10^4$ iterations (or simulated topics)
 }
 \label{fig:plot_models}
 \end{figure}
 
 \begin{table}[t]
 \caption{Fraction of pairs failing Chi-Square Goodness of Fit Test}
 \centering
 \begin{tabular}{|c|c|c|c|c|c|}
 \hline
 \T \B
 RPM-IC &  RPM-S & RPM-D & RPM-R & MSM-SP
 \\ \hline \T \B
 0.001 & 0.02 & 0.96 & 0.99 & NA
 \\ \hline
 \end{tabular}
 \end{table}

 The chi-square statistic is high (punishes heavily) if the deduced distribution allots non-negligible probability to values which are allotted negligible probability by the actual distribution. This was the case for a large fraction of pairs when using RPM-D and for over 99\% of pairs when using RPM-R, thus leading to high fractions of failures. 

 \begin{remark}
 Owing to a good balance between accuracy and efficiency, RPM-S can be justified to be used for deducing the preferences of nodes in a social network, with the knowledge of preferences of a subset of nodes. 
 However, if we need only the deduced edge similarities for all pairs of nodes without having to deduce the actual preferences, MSM-SP is a promising model.
 \end{remark}

 As explained in Section \ref{sec:overview}, a model that could deduce the preferences of all the nodes, would allow us to deduce the actual aggregate preference of the network. This would hence allow us to estimate how closely the aggregate preference obtained using a  set of representative nodes resembles the actual aggregate preference.
 With an estimate of closeness associated with every candidate representative set, it would now be possible to determine a representative set with the closest resemblance.
 The focus of the following sections will be on formulating an objective function that quantifies this closeness, and hence developing algorithms for determining the best representative set.


 \section{The Sampling-based Preference \\Aggregation Problem}
 \label{sec:problem_pasn}

 Given a network with a set of nodes $N$ and an aggregation rule $f$, our objective is to choose a set of representative nodes $M \subseteq N$ of certain cardinality $k$, and aggregate their preferences to arrive at an aggregate preference that is `close enough' to the aggregate preference of $N$ using $f$, in expectation (in expectation, owing to the stochastic nature of the edge similarities). 
 We now formalize this problem.

 Let the expected distance between a set $S \subseteq N$ and a node $i \in N$ be
 \begin{equation}
 \label{eqn:dist}
 d(S,i) = \min_{j \in S} d(j,i)
 \end{equation}
 We call $d(S,i)$ as the `expected' distance since $d(j,i)$ is the expected distance between nodes $j$ and $i$ with respect to their preferences.
 Since $d(i,i)=0,\forall i \in N$, we have $d(S,j)=0,\forall j \in S$.
 Let
 \begin{equation}
 \label{eqn:repr}
 \Phi(S,i) \sim_\mathcal{U}\argmin_{j \in S} d(j,i)
 \end{equation}
 be a node chosen uniformly at random from the set of nodes in $S$ that are closest in expectation to node $i$ in terms of preferences. We say that $\Phi(S,i)$ represents node $i$ in set $S$. In other words, $\Phi(S,i)$ is the unique {\em representative} of $i$ in $S$.

  \vspace{-1.5mm}
 \subsection{Aggregating Preferences of Representative Nodes}
 \label{sec:aggr_repr}

 Recall that preference profile is a multiset containing preferences of the nodes.
 Let the preference profile of the population $N$ be $P$ and that of the selected representative set $M$ be $Q$. 
 Suppose $M=\{i,j\}$ where
 $j$ represents, say ten nodes
  including itself, 
 while $i$ represents one node
 (only itself). 
 If the preferences are aggregated by 
 feeding $Q$ to aggregation rule $f$,
 the aggregate preference $f(Q)$ so obtained may not reflect the preferences of the population, in general, owing to the asymmetry in importance of the selected nodes.
 So
  to capture this asymmetry, their preferences must be weighted. 
 In our approach, the weight given to the preference of a node is precisely the number of nodes that it represents.

 Let $Q'$ be the preference profile obtained by replacing every node's preference in $P$ by its uniquely chosen representative's preference. So, $k=|M|=|Q| \leq |Q'|=|P|=|N|=n$.
 In our approach, the weight of a representative implies the number of nodes it represents or equivalently, the number of times its preference appears in the new preference profile.
 So in the above example, 
 the new  profile $Q'$ consists of ten preferences of $j$ and one of $i$.
 Thus we aggregate the preferences of selected nodes using $f(Q')$.

 So the problem under consideration can be viewed as a setting where given certain nodes representing a population, every node in the population is asked to choose one among them as its representative; now the representatives vote on behalf of the nodes who chose them.
 
  \vspace{-1.5mm}
 \subsection{A Measure of `Close Enough'}
 \label{sec:Delta}

 Now given $k$, our objective is to select a set of nodes $M$ such that $|M|=k$, who report their preferences such that, 
 in expectation, the error incurred in using the aggregate preference, say $f(R)$, obtained by aggregating the preferences of the nodes in $M$ (in an unweighted manner if $R=Q$ or in a weighted manner if $R=Q'$) instead of $f(P)$ obtained by aggregating the preferences of the nodes in $N$, is minimized.
 Note that an aggregation rule $f$ may not output a unique aggregate preference, that is, $f$ is a correspondence. 
 So
 the aggregation rule $f$ on the preferences of the entire population outputs $f(P)$ which is a set of preferences. 
 
 If $f(R)$ is a set of multiple preferences, we need to have a way to determine how close it is to $f(P)$. For this purpose, we propose an extension of Kendall-Tau distance for sets of preferences. Now, since $f(P)$ is generally not known and all preferences in $f(R)$ are equivalent in our view, we choose a preference from $f(R)$ uniformly at random and see how far we are from the actual aggregate preference, in expectation. 
 In order to claim that a chosen preference in $f(R)$ is a good approximation, it suffices to show that it is close to at least one preference in $f(P)$.  Also, as any preference $y$ in $f(R)$ is chosen uniformly at random, we define the error incurred in using $f(R)$ instead of $f(P)$ as
 \begin{equation}
 \label{eqn:Delta}
 f(P) \, \Delta \, f(R) = \mathbb{E}_{y \sim_\mathcal{U} f(R)} \left[ \min_{x \in f(P)} \tilde{d}(x,y) \right]
 \end{equation}
 where $\tilde{d}(x,y)$ is the distance between preferences $x$ and $y$ in terms of the same distance measure as $d(\cdot,\cdot)$ (normalized Kendall-Tau distance in our case). Notice that in general, $f(P) \, \Delta \, f(R) \neq f(R) \, \Delta \, f(P)$.
 For instance, if $f(P)=\{p_A,p_B\},f(R)=\{p_A\}$, the error is zero since we have obtained an aggregate preference that is among the actual aggregate preferences. On the other hand, 
 if $f(P)=\{p_A\},f(R)=\{p_A,p_B\}$, there is a half probability of choosing $p_A$ which is consistent with the actual aggregate preference, however, there is a half probability of choosing $p_B$ which is inconsistent; this results in $f(P) \, \Delta \, f(R) = \frac{1}{2}\,\tilde{d}(p_B,p_A)$.
 Note that $\Delta$ can be defined in several other ways depending on the application or the case we are interested in (worst, best, average, etc.). 
 In this paper, 
 we use the definition of $\Delta$ as given in Equation~(\ref{eqn:Delta}).
 
 Recall that the distance between a pair of nodes is drawn from distribution $\mathcal{D}$ with the corresponding parameters, so the realized values for different topics would be different in general. The value $f(P) \, \Delta \, f(R)$ can be obtained for every topic and hence the expected error $\mathbb{E} [ f(P) \, \Delta \, f(R) ]$ can be computed by averaging the values over all topics.
 It can be easily seen that $\mathbb{E} [ f(P) \, \Delta \, f(R) ] \in [0,1]$.
 Now our objective is to find a set $M$ such that $\mathbb{E} [ f(P) \, \Delta \, f(R) ]$ is minimized.

  \vspace{-1mm}
 \subsection{An Abstraction of the Problem}
 \label{sec:abstraction}

 For aggregation rule $f$, we define the objective function to be $\mathcal{F}(M) = 1- \mathbb{E} [ f(P)\;\Delta \;f(R) ]$ with the objective of finding a set $M$ that maximizes this value. 
 However, even if $M$ is given, computing $\mathcal{F}(M)$ is computationally intensive for several aggregation rules and furthermore, hard for rules such as Kemeny. 
 It can be seen that $\mathcal{F}(\cdot)$ is not monotone for non-dictatorial aggregation rules (the reader is referred to Figure~\ref{fig:plots_pasn} for the non-monotonic plots of Greedy-sum and Degree-cen algorithms since
 in a run of these algorithms, a set of certain cardinality is a superset of any set having a smaller cardinality).
 It can also be  checked empirically
 that $\mathcal{F}(\cdot)$ is neither submodular nor supermodular.
 Even for simple non-dictatorial aggregation rules, it is not clear if one could efficiently find a set $M$ that maximizes $\mathcal{F}(\cdot)$, within any constant approximation factor.
 This motivates us to propose an approach that finds set $M$ \textit{agnostic to the aggregation rule} being used.
 To this end, we propose a property for preference aggregation rules, {\em weak insensitivity}.
 \begin{definition}[weak insensitivity property]
 \label{def:weakins}
 A preference aggregation rule satisfies {\em weak insensitivity property\/} under a distance measure and an error operator between aggregate preferences $\Delta$, if and only if for any $\epsilon_d$, a change of $\eta_i \leq \epsilon_d$ in the preferences of all $i$, results in a change of at most $\epsilon_d$ in the aggregate preference.
 That is, $\forall \epsilon_d,$
 \begin{equation}
 \nonumber
 \;\eta_i \leq \epsilon_d \;,\; \forall i \in N \implies f(P) \, \Delta \, f(P')  \leq \epsilon_d
 \end{equation}
 where $P'$ is the preference profile of voters after deviations.
 \end{definition}
 
 We call it `weak' insensitivity property because it allows `limited' change in the aggregate preference (strong insensitivity can be thought of as a property that allows no change). 
 This is potentially an important property that an aggregation rule should satisfy as it is a measure of its robustness in some sense.
 It is clear that under normalized Kendall-Tau distance measure and $\Delta$ as defined in Equation~(\ref{eqn:Delta}), an aggregation rule that outputs a random preference does not satisfy weak insensitivity property as it fails the criterion for any $\epsilon_d < 1$, whereas dictatorship rule that outputs the preference of a single voter satisfies the property trivially.
 For our purpose, we propose a weaker form of this property, which we call {\em expected weak insensitivity}.
 
 \begin{definition}[expected weak insensitivity property]
 \label{def:gaussweakins}
 A preference aggregation rule satisfies {\em expected weak insensitivity property\/} under a distribution, a distance measure, and an error operator between aggregate preferences $\Delta$, if and only if for any $\mu_d$, a change of $\eta_i$ in the preferences of all $i$, where $\eta_i$ is drawn from the distribution with mean $\delta_i \leq \mu_d$ and any permissible standard deviation $\sigma_d$, results in an expected change of at most $\mu_d$ in the aggregate preference.
 That is, $\forall \mu_d, \; \forall \text{ permissible } \sigma_d ,  $
 \begin{equation}
 \label{eqn:weak_ins}
 \;\delta_i \leq \mu_d \;,\; \forall i \in N \implies \mathbb{E} [ f(P) \, \Delta \, f(P') ] \leq \mu_d
 \end{equation}
 where $P'$ is the preference profile of voters after deviations.
 \end{definition}
 Note that in $\mathbb{E} [ f(P) \, \Delta \, f(P') ]$, the expectation is over the varying modified preferences of the voters (since $\eta_i$'s vary across instances and also, there are multiple preferences at a distance of $\eta_i$ from any given preference, in general). 
 In this paper, we study expected weak insensitivity property under  distribution $\mathcal{D}$, normalized Kendall-Tau distance, and $\Delta$ as defined in Equation~(\ref{eqn:Delta}). 
 For distribution $\mathcal{D}$ with $\mu_d \in [0,1]$,
 the permissible range of $\sigma_d$ depends on $\mu_d$. 
 This range is wider for intermediate values of $\mu_d$ and shortens as we move towards the extremes.
 In any case,
 the permissible range for $\sigma_d$ cannot exceed $\frac{1}{\sqrt{12}} \approx 0.28$ (value at which the truncated Gaussian becomes a Uniform distribution), while for $\mu_d \in \{0,1\}$, the permissible $\sigma_d=0$ (since a Gaussian distribution truncated in $[0,1]$ with any non-zero standard deviation cannot have 0 or 1 as the mean).

 \begin{remark}
 We conducted extensive simulations for investigating empirical satisfaction of the expected weak insensitivity property by the considered aggregation rules under  distribution $\mathcal{D}$, normalized Kendall-Tau distance, and $\Delta$ as defined in Equation~(\ref{eqn:Delta}). The simulations considered values of $r=3,\ldots,7$, $n = 100,200,\ldots,1000$, different distributions on the voter preferences, discrete values of $\mu_d$ separated by $1/$\begin{scriptsize}$\dbinom{r}{2}$\end{scriptsize} (the finest resolution for the given $r$), and discrete values of $\sigma_d$ separated by 10\% of the maximum permissible value for the corresponding $\mu_d$. 
 For any fixed preference profile,
 Table \ref{tab:weak_ins_rules} presents 
 typical fractions of simulation runs in which Criterion (\ref{eqn:weak_ins}) was satisfied, by various aggregation rules.
 Dictatorship rule always satisfied the criterion, whereas Smith set and Schulze rules satisfied it almost always.
 It was, however, violated by Veto rule for a large fraction of simulation runs.
 For all other considered rules, if (\ref{eqn:weak_ins}) was violated, it was usually for the lowest values of $\mu_d$.
 Furthermore, in most cases, the extent of violation was not very significant.
 So if $\mu_d$ is not very small, these rules could be assumed to satisfy the expected weak insensitivity property for practical purposes.
 \end{remark}
 
 \begin{table}[t]
 \caption{
 Typical fractions of simulation runs in which the criterion for expected weak insensitivity was satisfied by various aggregation rules
 under distribution $\mathcal{D}$, normalized Kendall-Tau distance, and $\Delta$ as defined
 }
 \begin{center}
 \begin{tabular}{|>{\small}p{1.74cm}|>{\small}c||>{\small}p{1.48cm}|>{\small}c||>{\small}p{1.2cm}|>{\small}c|}
 \hline
 \T \B 
 Dictatorship & 1.00 & Smith set & .998 & Schulze & .997   \\ \hline
 \T \B 
 Copeland & .97  & Plurality & .95 & Borda & .92 \\ \hline 
 \end{tabular}
 \begin{tabular}{|>{\small}p{1.05cm}|>{\small}c||>{\small}p{1.05cm}|>{\small}c||>{\small}p{.95cm}|>{\small}c||>{\small}p{.5cm}|>{\small}c|}
 \T \B 
 Minmax & .87 & Kemeny & .85 & Bucklin & .82   & Veto & .69 \\ \hline
 \end{tabular}
 \end{center}
 \label{tab:weak_ins_rules}
   \vspace{-1.5mm}
 \end{table}
 
 \begin{lemma}
 \label{lem:guarantee}
 Given a distance measure and a $\Delta$, with a preference aggregation rule satisfying expected weak insensitivity property under  distribution $\mathcal{D}$,
 if the expected distance between every node and 
 the representative set $M$ is at most $\epsilon_d \in [0,1]$,
 then the expected error incurred in using $f(Q')$ instead of $f(P)$ is at most $\epsilon_d$. 
 That is, for $\epsilon_d \in [0,1] , \; $
 \begin{displaymath}
 d(M,i) \leq \epsilon_d \;,\; \forall i \in N \implies \mathbb{E} [ f(P) \, \Delta \, f(Q') ] \leq \epsilon_d
 \end{displaymath}
 \end{lemma}
 \begin{proof}
 In the preference profile $P$ of all nodes, the preference of any node $i \in N$ is replaced by the preference of its representative node $p=\Phi(M,i)$
 to obtain $Q'$. 
 From Equations~(\ref{eqn:dist}), (\ref{eqn:repr}), and the hypothesis, we have $d(p,i) \leq \epsilon_d$.
 
 Since in $P$, preference of every $i$ is replaced by that of the corresponding $p$ to obtain $Q'$, and distance between $i$ and $p$ is distributed according to  distribution $\mathcal{D}$ with mean $d(p,i)$
 and some standard deviation $\sigma_d$, 
 the above is equivalent to node $i$ deviating its preference by some value which is drawn from distribution $\mathcal{D}$ with mean $d(p,i) = d(M,i)$. 
 So we can map these variables to the corresponding variables in Equation~(\ref{eqn:weak_ins}) as follows: $\delta_i = d(M,i) \; \forall i$, $\mu_d = \epsilon_d$, and $P'=Q'$.
 Also, recall that in $\mathbb{E} [ f(P) \, \Delta \, f(P') ]$, the expectation is over varying modified preferences of the nodes, while in $\mathbb{E} [ f(P) \, \Delta \, f(Q') ]$, the expectation is over varying preferences of the nodes' representatives in $M$ with respect to different topics (and hence preferences) of the nodes.
 These are equivalent given $P'=Q'$.
 As this argument is valid for any permissible $\sigma_d$, the result follows.
 \end{proof}
 
 So under the proposed model and for aggregation rules satisfying the expected weak insensitivity property, Lemma~\ref{lem:guarantee} establishes a 
  relation between (a) the closeness of the chosen representative set to the population in terms of expected distance and (b) the error incurred in the aggregate preference if that set is chosen as the representative set.
 We now return to our goal of abstracting the problem of determining a representative set, by proposing an approach that is agnostic to the aggregation rule being used.

 \subsection{\mbox{Objective Functions in the Abstracted Problem}}
 \label{sec:objfn_abstract}
 
 Recall that $c(\cdot,\cdot) = 1 - d(\cdot,\cdot)$. 
 Our objective is now to find a set of critical nodes $M$ that maximizes a certain objective function, with the hope of minimizing  $\mathbb{E} [ f(P) \, \Delta \, f(R) ]$ where $R=Q'$ in our case.
 As the aggregation rule is anonymous, in order to ensure that the approach works well, even for rules such as random dictatorship, the worst-case objective function for the problem under consideration, representing least expected similarity, is
 \begin{equation}
 \label{eqn:influence_min}
 \rho(S) = \min_{i \in N} c(S,i)
 \end{equation}
 The above is equivalent to: $\max_{i \in N} d(S,i) = 1-\rho(S)$.
 Thus $\epsilon_d = 1-\rho(S)$ in Lemma~\ref{lem:guarantee}, and so this objective function offers a guarantee on $\mathbb{E} [ f(P) \, \Delta \, f(Q') ]$
 for aggregation rules satisfying the expected weak insensitivity property.
 We will provide a detailed analysis for the performance guarantee of an algorithm that aims to maximize $\rho(S)$, in Section~\ref{sec:greedymin_guarantee}.
 
 Now the above worst-case objective function ensures that our approach works well even for aggregation rules such as random dictatorship.
 However, such extreme aggregation rules are seldom used in real-world scenarios; hence, another surrogate objective function, representing average expected similarity, or equivalently sum of expected similarities, is
 \begin{equation}
 \label{eqn:influence_sum}
 \psi(S) = \sum_{i \in N} c(S,i)
 \end{equation}
         We will look into the desirable properties of an algorithm that aims to maximize $\psi(S)$, in Section~\ref{sec:cooperativeview}.
 
 We now turn towards the problem of maximizing the above two surrogate objective functions.
 \begin{proposition}
 \label{prop:nphard}
 Given constants $\chi$ and $\omega$, \\
 (a) it is NP-hard to determine whether there exists a set $M$ consisting of $k$ nodes such that $\rho(M) \geq \chi$, and \\
 (b) it is NP-hard to determine whether there exists a set $M$ consisting of $k$ nodes such that $\psi(M) \geq \omega$.
 \end{proposition}
 
 We provide a proof of Proposition~\ref{prop:nphard} in Appendix \ref{app:nphard_proof}. 
 A function $h(\cdot)$ is said to be {\em submodular} if,
 for all $v \in N \setminus T$ and for all $S,T$ such that $S \subset T \subset N$,
 \begin{equation}
 \label{eqn:submodular}
 \nonumber
 h(S \cup \{v\}) - h(S) \geq h(T \cup \{v\}) - h(T)
 \end{equation}

 \begin{proposition}
 \label{prop:submodularity}
 The objective functions $\rho(\cdot)$ and $\psi(\cdot)$ 
 are non-negative, monotone increasing, and submodular.
 \end{proposition}

 We provide a proof of Proposition~\ref{prop:submodularity} in Appendix \ref{app:submod_proof}. 
 For a non-negative, monotone increasing, and submodular function, the greedy hill-climbing algorithm (selecting elements one at a time, each time choosing an element that provides the largest marginal increase in the function value), gives a $(1-\frac{1}{e}) \approx 0.63$-approximation to the optimal solution~\cite{nemhauser1978analysis}.
 As the considered objective functions in Equations~(\ref{eqn:influence_min}) and (\ref{eqn:influence_sum}) satisfy these properties, we use the greedy hill-climbing algorithm to obtain a good approximation to the optimal solution.
 Moreover, as desired, the functions are agnostic to the aggregation rule being used.
 
 We next devise algorithms for finding a representative set, present their performance with the aid of extensive experimentation, and provide detailed analysis of the results.
 
 
 \section{Selection of the Representative Set: \\Algorithms and Performance}
 \label{sec:results}
 
 Recall that the preference profile of $N$ is $P$, that of $M$ is $Q$, and that obtained by replacing every node's preference in $P$ by that of its uniquely chosen representative in $M$, is $Q'$.
 Given the number of nodes to be selected $k$, our objective is to find a set $M$ of size $k$ such that 
 $\mathbb{E} [ f(P) \, \Delta \, f(R) ]$ is minimized, where $R=Q'$ or $Q$ depending on the algorithm.

 \subsection{Algorithms for Finding Representatives}
 \label{sec:algos}
 
 We now describe the algorithms we consider in our study.

 \begin{itemize}[\setlength{\labelwidth}{\widthof{\textbullet}}
 \setlength{\labelsep}{7.5pt}
  \setlength{\IEEElabelindent}{0pt}
     \IEEEiedlabeljustifyl ]
 \setlength\itemsep{.5em}
 \item \textbf{Greedy-orig} (Greedy hill-climbing for maximizing $1-\mathbb{E} [ f(P) \, \Delta \, f(Q') ]$):
 Initialize $M$ to $\{\}$.
 Until $|M|=k$, choose a node $j \in N \setminus M$ that 
 minimizes the expected error or equivalently, maximizes $1-\mathbb{E} [ f(P) \, \Delta \, f( {Q}_M') ]$,
 where $Q_M'$ is the preference profile obtained by replacing every node's preference in $P$ by the preference of its uniquely chosen representative in $M$.
 Note that the optimal set would depend on the aggregation rule $f$.
 Its time complexity for obtaining $M$ and hence $R$ is $O(kn\mathcal{T}_f)$, where $\mathcal{T}_f$ is the time complexity of obtaining an aggregate preference using the aggregation rule $f$. For instance, $\mathcal{T}_f$ is $O(rn)$ for plurality, $O(1)$ for dictatorship.

 \item \textbf{Greedy-sum} (Greedy hill-climbing for maximizing $\psi(\cdot)$):
 Initialize $M$ to $\{\}$.
 Until $|M|=k$, choose a node $j \in N \setminus M$ that 
 maximizes $\psi(M \cup \{j\}) - \psi(M)$.
 Then obtain $f(R) = f(Q')$.
 If the similarity matrix is known, its time complexity for obtaining $M$ and hence $R$ is $O(kn^2)$.
 If the similarity matrix is unknown, the time complexity for deriving it is largely decided by the model used for deducing the mean distances between all pairs of nodes. 
 
 \item \textbf{Greedy-min} (Greedy hill-climbing for maximizing $\rho(\cdot)$):
 Similar to Greedy-sum, with $\rho(\cdot)$ instead of $\psi(\cdot)$.
 
 \item \textbf{Between-cen} (Betweenness centrality heuristic): Choose $k$ nodes having the maximum values of Freeman's betweenness centrality (edge weights being dissimilarities).
 Then obtain $f(R) = f(Q)$.
 Its time complexity for obtaining $M$ is $O(nm + n^2 \log n)$.
 
 \item \textbf{Degree-cen} (Degree centrality heuristic): Choose $k$ nodes having the maximum weighted degrees (edge weights being similarities).
 Then obtain $f(R) = f(Q)$.
 Its time complexity for obtaining $M$ is $O(nk+n\log n)$.
 \item \textbf{Random-poll} (Random polling):
 Choose $k$ nodes uniformly at random.
 Then obtain $f(R) = f(Q)$.
 It is  an important baseline, since it is
 the most employed method in practice.
 Also, it has been claimed in the literature that it is optimal to ignore the social network \cite{conitzer2012should} and that random node selection performs well in practice \cite{leskovec2006sampling}.
 
 \item Other centrality measures 
 (PageRank, Katz, eigenvector):  
 In a weighted network like the one under study, the sum of weights of edges adjacent on a node could exceed 1.
 So the edge weights are required to be attenuated for computing these measures.
 For instance, while employing PageRank, the edges need to be converted to directed edges \cite{chen2010scalable} (perhaps by adding self loops as well) and reweighed so that the outgoing edge weights from any node, sum to 1.
 Katz centrality involves a parameter $\alpha$ which is required to be less than the reciprocal of the largest eigenvalue of the weighted adjacency matrix, thus effectively reweighing the weights of the walks. 
 \end{itemize}
 
 In all algorithms, 
 the time complexity of computing $f(R)$ depends on the aggregation rule $f$.
 For dictatorship, if the dictator is not in $M$, Random-poll outputs the preference of a node in $M$ chosen uniformly at random, else it outputs the dictator's preference; other algorithms output the preference of the dictator's representative in $M$.
 We now present 
  desirable properties of Greedy-min and Greedy-sum algorithms.

 \subsection{Performance Guarantee for Greedy-min Algorithm}
 \label{sec:greedymin_guarantee}
 
 Here, we show the performance guarantee of Greedy-min.
 \begin{theorem}
 \label{thm:approx_guarantee}
 For an aggregation rule satisfying expected weak insensitivity, the error incurred in using the aggregate preference given by the Greedy-min algorithm instead of the actual aggregate preference, is at most $\left(1 - \left(1-\frac{1}{e}\right) \rho^*\right)$, where $\rho^* = \max_{S \subseteq N,|S|\leq k} \rho(S)$.
 \end{theorem}
 \begin{proof}
 Let $S^G$ be a set obtained using greedy hill-climbing algorithm for maximizing $\rho(\cdot)$. Since greedy hill-climbing provides a $\left(1-\frac{1}{e}\right)$-approximation to the optimal solution, we have
 
 \vspace{-3mm}
 \begin{small}
 \begin{align*}
  &\;\;\;\;\;\;\;\;\;\;
  \rho(S^G) = \min_{i \in N} c(S^G,i) \geq \left(1-\frac{1}{e}\right) \rho^* \\
  &\implies
   1 - \max_{i \in N} d(S^G,i) \geq \left(1-\frac{1}{e}\right) \rho^* \\
 &\implies
 \max_{i \in N} d(S^G,i) \leq  1-\left(1-\frac{1}{e}\right) \rho^*\\
 &\implies
 d(S^G,i) \leq  1-\left(1-\frac{1}{e}\right) \rho^*, \;\;\; \forall i\in N
 \end{align*}
 \end{small}

 For an aggregation rule satisfying expected weak insensitivity property, from Lemma~\ref{lem:guarantee}, when the representative set $M=S^G$, we have
 $ 
 \nonumber
 \mathbb{E} [ f(P) \, \Delta \, f(Q') ] \leq 1 - \left(1-\frac{1}{e}\right) \rho^*
 $. 
 \end{proof}
 
 It is to be noted that, though the approximation ratio given by the greedy algorithm is modest in theory, it has been observed in several domains that its performance is close to optimal in practice when it comes to optimizing non-negative, monotone increasing, submodular functions.

 \subsection{A Cooperative Game Theoretic Viewpoint of Greedy-sum Algorithm}
 \label{sec:cooperativeview}

 Shapley value is known to act as a good measure for node selection problems in social networks, particularly that of influence maximization \cite{narayanam2011shapley,michalak2013efficient}.
 We have seen that, to maximize the objective function
 $
 \psi(S) = \sum_{i \in N} c(S,i)
 $,
 the greedy hill-climbing algorithm first chooses a node $j$ that maximizes
 $\sum_{i \in N} c(i,j)$ or equivalently $\sum_{i \in N, i\neq j} c(i,j)$ (since $c(j,j)=1, \forall j$).
 It has been shown in \cite{garg2013novel} that $\sum_{i \in N, i\neq j} c(i,j)$ is the Shapley value of player $j$, in a convex Transferable Utility (TU) game $(N,\nu)$ with the characteristic function 
 $ 
 \nu(S) = \sum_{{i,j \in S, i \neq j}} c(i,j) 
 $. 
 This characteristic function can be viewed as an indication of how tightly knit a group is, or how similar the members of a set $S$ are to each other.
 Let $\phi(\nu),Nu(\nu),Gv(\nu),\tau(\nu)$ be Shapley value, Nucleolus, Gately point, $\tau$-value of the TU game $(N,\nu)$.

 \begin{theorem} 
 \label{shapeqnu}
  For the TU game defined by  
  $\nu(S) = \sum_{{i,j \in S, i \neq j}} c(i,j) $,
 $ 
 \nonumber
  \phi(\nu)=Nu(\nu)=Gv(\nu)=\tau(\nu).
 $ 
 \end{theorem}
 We provide a proof of Theorem~\ref{shapeqnu} in Appendix \ref{app:shapeqnu}. 
 So the Greedy-sum algorithm aims to maximize a term that is unanimously suggested by several solution concepts for a TU game capturing the similarities within a set.
 In other words, the solution concepts unanimously suggest that the first node chosen by the Greedy-sum algorithm, is 
 on average most similar to different subsets of the population.

 \subsection{Experimental Observations}
 \label{sec:exp}

 After obtaining the representative set using the aforementioned algorithms, we tested their performance on $\mathbb{T}=10^4$ topics or preference profiles generated using the RPM-S model (with the assigned neighbor chosen in a random way) on our Facebook data. 
 Owing to the nature of the Random-poll algorithm, we ran it sufficient number of times to get an independent representative set each time, and then defined the performance as the average over all the runs.
 The values of $\mathbb{E} [ f(P) \, \Delta \, f(R) ]$ were computed using extensive simulations with the considered aggregation rules. 

 We observed that for any node $i$, in all algorithms except Random-poll, the candidate set $\argmin_{j \in S} d(j,i)$ (see Equation (\ref{eqn:repr})) was a singleton for low values of $k$. For higher values of $k$ (higher than 7), the number of candidates were usually less than or equal to 4, with a maximum of 10 for $k=50$ in one instance. We thus ran the experiments several times; we observed that changes in the error plots were not very significant.
 Also, the aggregate preferences $f(P)$ and $f(R)$ both consisted of one preference in almost all the runs, so the error
 $\mathbb{E} [ f(P) \, \Delta \, f(R) ]$ was equivalent to  Kendall-Tau distance between the actual aggregate preference and the obtained aggregate preference (see Equation (\ref{eqn:Delta})).

 Figure \ref{fig:extra}(a) shows the plots for the worst case of Random Dictatorship, that is, when the randomly chosen dictator is the most dissimilar to the chosen representative set. 
 Greedy-min performed the best owing to it ensuring that no node in the network is very dissimilar to the chosen representative set.
 The error for Greedy-orig could not be plotted since the objective function cannot be computed in this case.

 The plots for all non-dictatorial aggregation rules were similar (albeit with different scaling) to the ones plotted in Figure~\ref{fig:plots_pasn}. 
 Our key observations are as follows:
 
 \begin{figure}
 \centering
 \includegraphics[scale=0.63]{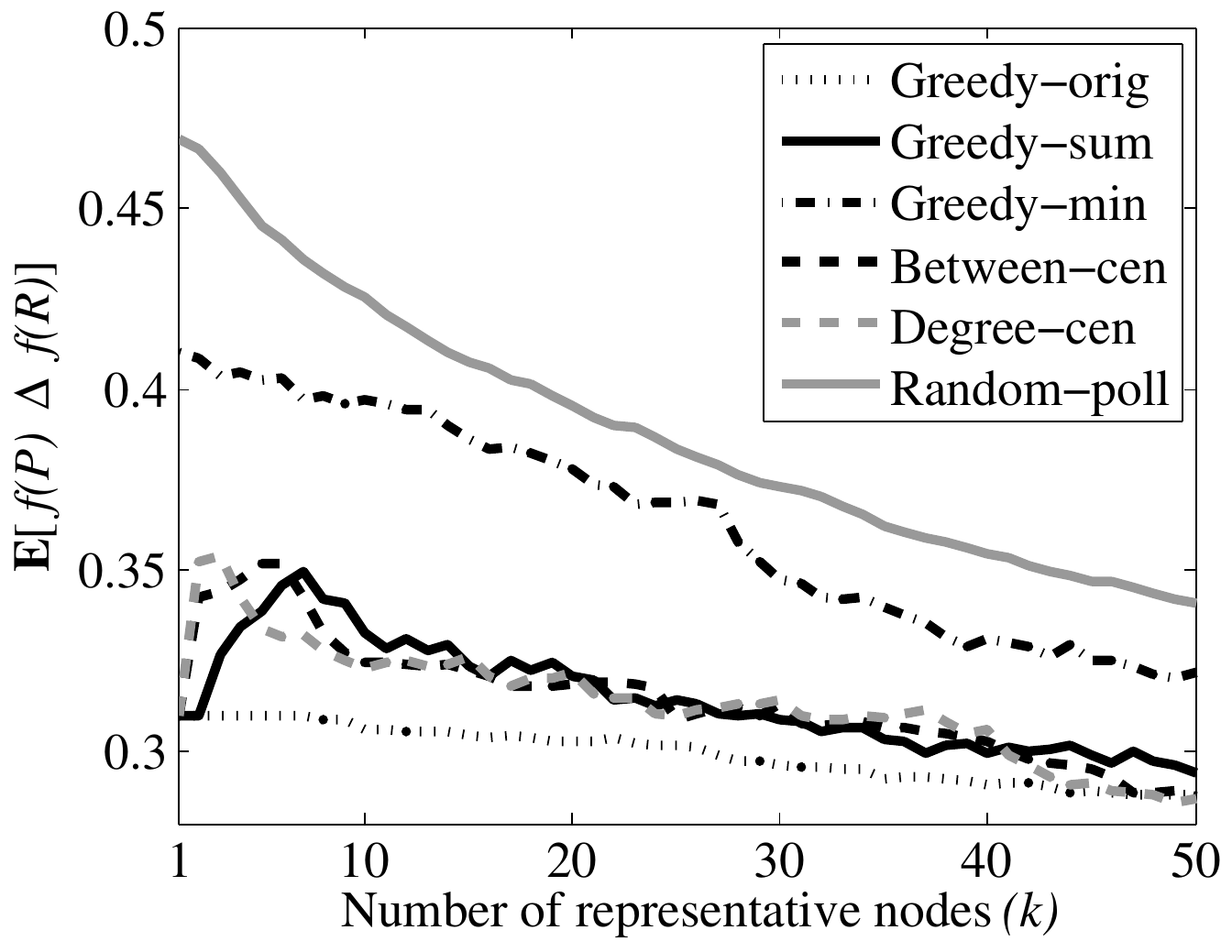} 
    \vspace{2mm}
 \caption{
 Comparison among algorithms for Plurality aggregation rule
 }
 \label{fig:plots_pasn}
    \vspace{2mm}
 \end{figure}

 \begin{itemize}[\setlength{\labelwidth}{\widthof{\textbullet}}
 \setlength{\labelsep}{7.5pt}
  \setlength{\IEEElabelindent}{0pt}
     \IEEEiedlabeljustifyl ]
     
     \setlength\itemsep{.5em}
 
   \vspace{5mm}
 \item 
 \textbf{Performance of Greedy-orig.}
 Greedy-orig performed the best throughout, however it had unacceptably high running time (order of days) even for computationally fast aggregation rules such as plurality; so it is practically infeasible to run this algorithm for computationally intensive rules, for example, Kemeny.
 
    \vspace{3mm}
 \item 
 \textbf{Performance of Greedy-sum.}
 Greedy-sum performed very well;
 but its plots displayed non-monotonicity especially in the lower range of $k$, and so a higher $k$ might not always lead to a better result.
 Greedy-sum, with MSM-SP as a precursor, attempts to find nodes which are closest to other nodes on average, where closeness is based on the similarity deduced using adaptation of the shortest path algorithm (as described in Section \ref{sec:msm_sp}). 
 This is on similar lines as finding influential nodes for diffusing information, which would lead to other nodes being influenced with maximum probabilities on average.
 
    \vspace{3mm}
 \item 
 \textbf{Performance of Greedy-min.}
 Greedy-min performed better than Random-poll for low values of $k$; this difference in performance decreased for higher values of $k$.
 The effect of satisfaction or otherwise of expected weak insensitivity was not very prominent, because the property is not violated by an appreciable enough margin for any aggregation rule.
 Nonetheless, the expected weak insensitivity property does provide a guarantee on the performance of Greedy-min for an aggregation rule.
 
    \vspace{3mm}
 \item 
 \textbf{Performance of Random-poll.}
 As mentioned earlier, the performance of Random-poll is based on an average over several runs; the variance in performance was very high for low values of $k$, and the worst case performance was unacceptable. The variance was acceptable for higher values of $k$.
 One reason for its performance being not very bad on average could be the low standard deviation of the mean distances (see Section \ref{sec:model_pairs}).
 
    \vspace{3mm}
 \item 
 \textbf{Performance of Between-cen.}
 Between-cen lagged behind Greedy-sum when the size of the representative set was small. However, it performed at par with or at times better  when the size of the representative set was moderate to high.
 Both Between-cen and Greedy-sum are based on the idea of shortest paths, albeit with different additive operators (simple addition in case of Between-cen versus MSM-SP in case of Greedy-sum) and with the difference that Between-cen concerns intermediary nodes and Greedy-sum concerns end nodes. 
 
    \vspace{3mm}
 \item 
 \textbf{Performance of Degree-cen.}
 Degree-cen showed a perfect balance between performance and running time. This demonstrates that high degree nodes indeed serve as good representatives of the population.

    \vspace{3mm}
 \item
 \textbf{Performance of other centrality measures.}
 Centrality measures (not included in the plots) such as Katz and eigenvector (belonging to the Bonanich family) as well as PageRank performed almost at par with Between-cen for most values of $k$. 
 Katz and eigenvector centralities, in particular, selected representative nodes which often were neighbors of each other; this is intuitively undesirable for the problem of sampling representatives which should ideally represent different parts of the network.
 Also, as discussed earlier, these measures reweigh edges,
 which may be detrimental to their performance.

    \vspace{3mm}
 \item 
 \textbf{Non-monotonicity of error plots.}
 The error plot need not be monotone decreasing with the representative set, that is, adding a node to a representative set need not reduce the error. For example, let $s_1$ be the selected representative node for $k=1$, and $s_2$ be the node added when $k=2$. If we aggregate preferences of the representative nodes in an unweighted way, it is clear that if $s_1$ is truly a good representative of almost the entire population and $s_2$ is a good representative of only a section of the population, weighing their preferences equally ($R=Q$) would lead to more error than considering $s_1$ alone (with $k=1$). This is precisely the reason why we employed the method of weighing their preferences differently while aggregating their preferences.  Weighing their preferences proportional to the number of nodes they represent ($R=Q'$) based on the deduced similarity matrix, was used as only a heuristic and so does not guarantee that the error would reduce with an increasing $k$. 
 
    \vspace{3mm}
 \item 
 \textbf{Role of social network.}
 We observe that algorithms which consider the underlying social network perform better than random polling, implying that the network plays a role and should be considered while determining representatives. 
 We next provide a more detailed insight.

 \end{itemize}

  \begin{figure*}
  \begin{tabular}{ccc}
  \hspace{-3mm}
  \includegraphics[scale=.45]{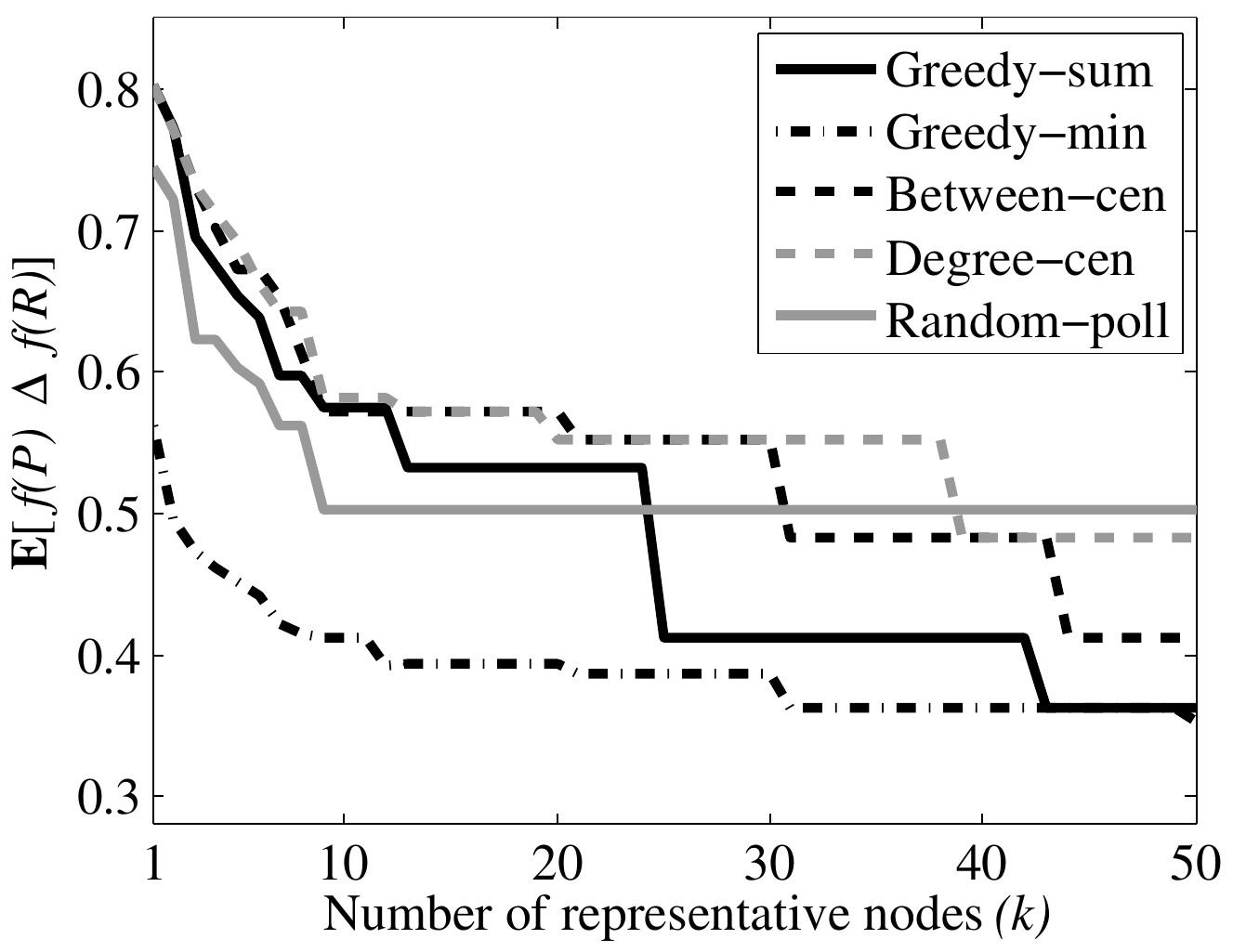}
  &
  \hspace{-5mm}
  \includegraphics[scale=.45]{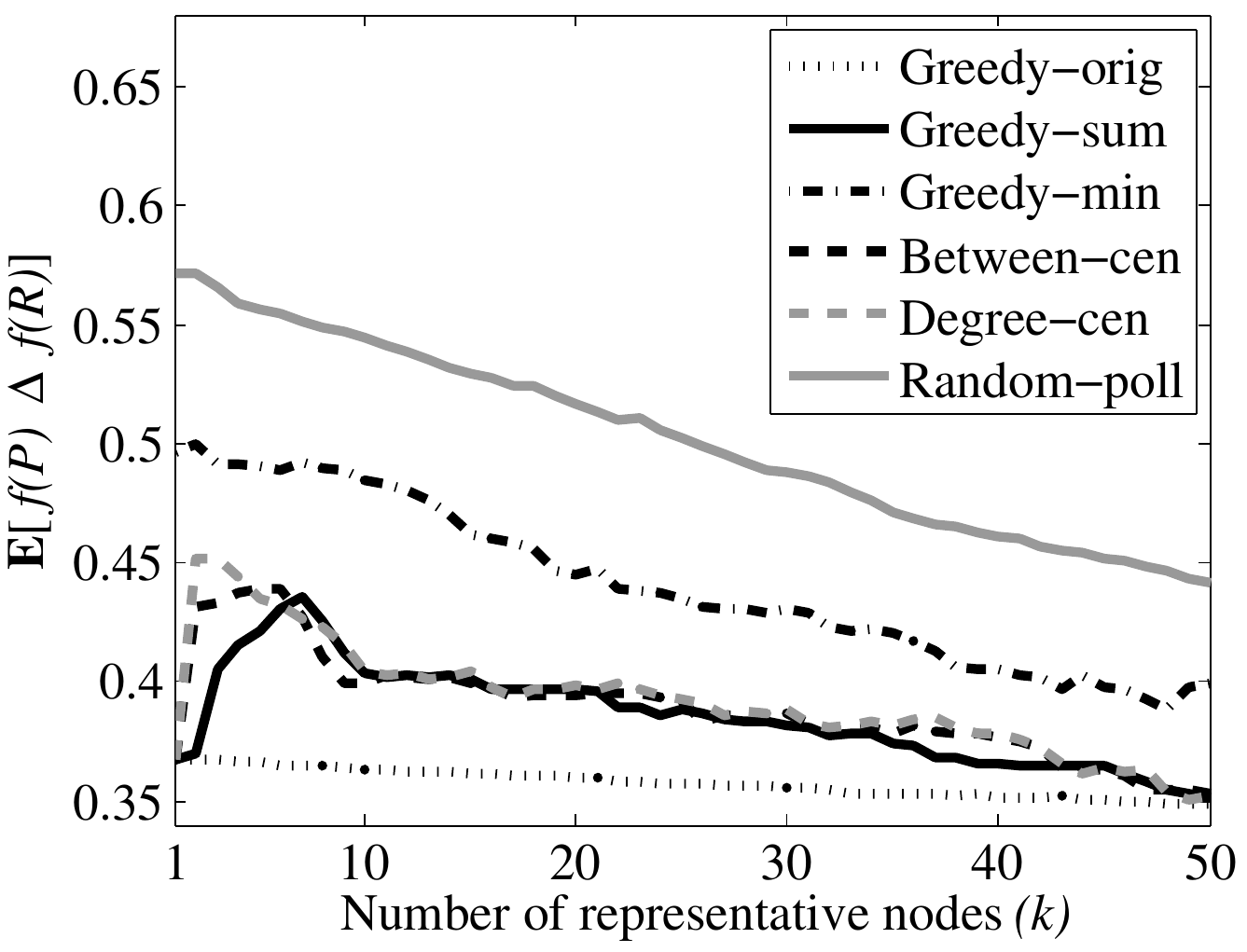}
  &
  \hspace{-5mm}
  \includegraphics[scale=.45]{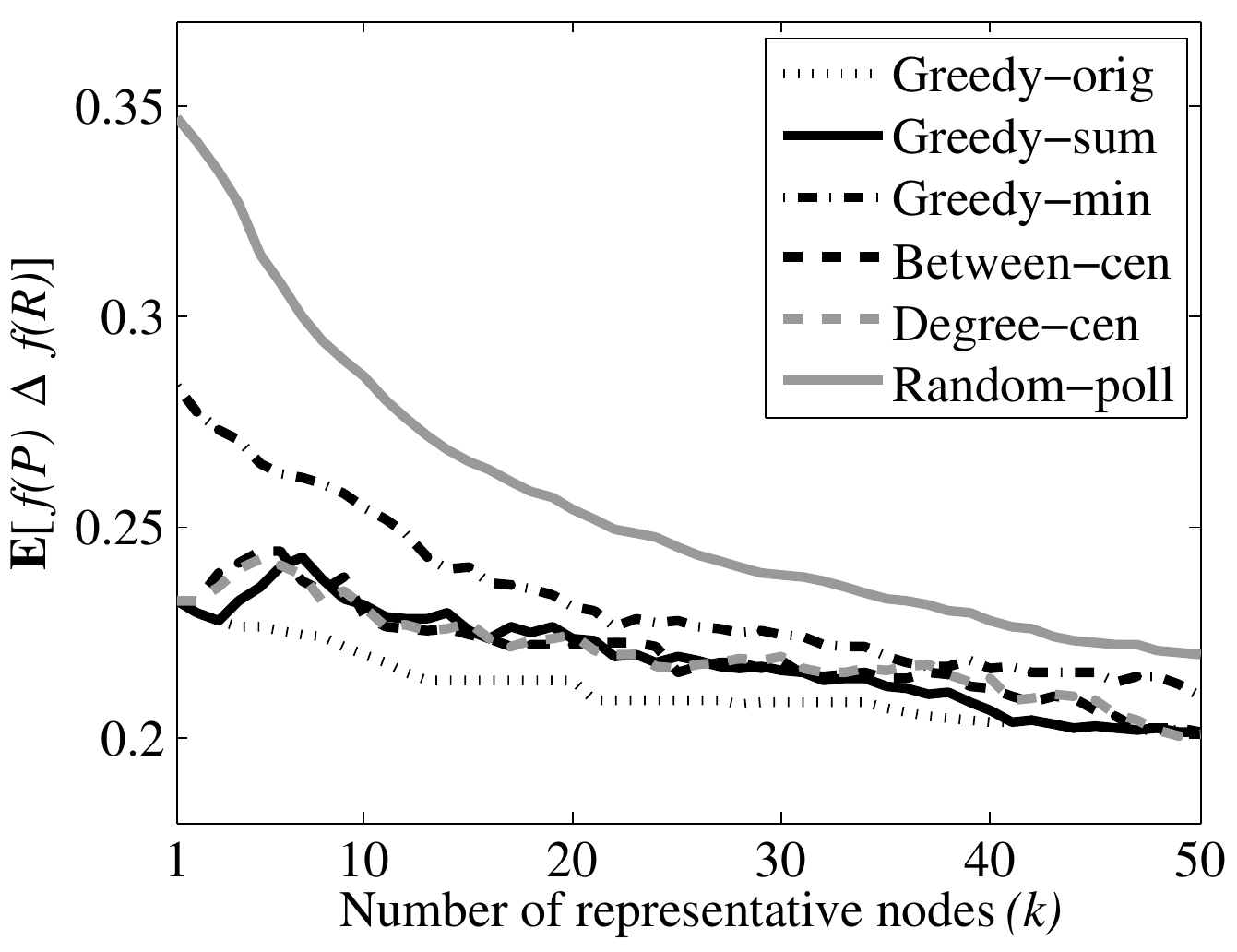}
  \\
  (a)  Random Dictatorship (worst case) 
  &
  (b) Plurality (Personal topics) 
  &
  (c) Plurality (Social topics) 
  \end{tabular}
  \caption{
  Comparison among algorithms 
  under other settings
  }
  \label{fig:extra}
  \end{figure*}

 \subsection{Personal versus Social Topics}
 \label{sec:pvss}

 We have focused on aggregating preferences across all topics, without classifying them into different types. 
 Now, we provide a preliminary analysis of personal versus social type of topics.
 It is to be noted that the model fitting (as discussed in Section \ref{sec:model_pairs}) was employed based on only four topics of each type that were available from our Facebook data. So the obtained results (Figures \ref{fig:extra}(b-c)) are provided with low confidence, but have qualitative implications.

 All the algorithms showed relatively low errors for social topics.
 Random-poll performed very well on average for sufficiently high values of $k$.
 This could be attributed to a lower average mean distance for social topics (0.30), resulting in nodes being more similar to each other; so nodes chosen at random are likely to be more similar to most other nodes.
 However, it performed badly while aggregating preferences with respect to personal topics, which has a higher standard deviation (0.12)~and also a higher average mean distance (0.40). Also, the variance in performance of Random-poll was unacceptably high for lower values of $k$, but acceptable for higher values. 
 So though the usage of Random-poll seems undesirable in general, 
 it is justified for social topics with a reasonable sample size.

 A high level of similarity between unconnected nodes with respect to social topics could be attributed to the impact of news and other common channels. It may also be justified by a theory of political communication~\cite{huckfeldt1995political} which stresses the importance of citizen discussion beyond the boundaries of cohesive groups for the dissemination of public opinion.

 
 \section{Conclusion}
 \label{sec:conclusion_pasn}
 
 This paper focused on two subproblems with respect to preference aggregation in social networks, namely, (a) how preferences are spread and (b) how to determine the best set of representative nodes.
 We started by motivating both these problems. 
 Based on our Facebook dataset,
  we developed a number of simple and natural models, of which RPM-S showed a good balance between accuracy and running time; while MSM-SP was observed to be the best model when our objective was to deduce the mean distances between all pairs of nodes, and not the preferences themselves.
 
 We formulated an objective function for representative-set selection and followed it up with two surrogate objective functions for practical usage. 
 We then proposed algorithms for selecting best representatives, wherein we provided a guarantee on the performance of the Greedy-min algorithm, subject to the aggregation rule satisfying the expected weak insensitivity property; we also studied the desirable properties of the Greedy-sum algorithm.
 We also observed that degree centrality heuristic performed very well, thus showing the ability of high-degree nodes to serve as good representatives of the population. 
 Our preliminary analysis also suggested that selecting representatives based on social network is advantageous for aggregating preferences related to personal topics, while random polling with a reasonable sample size is good enough for aggregating preferences related to social topics.

 \subsection{Future Work}
 \label{sec:future}

 It is intuitive that the network structure would affect the ease of sampling best representatives and how well they represent the population. 
 For instance, a tightly-knit network would likely consist of nodes with similar preferences, thus requiring a small number of representatives,
 while
 a sparse network would require a higher number.
 A network with a larger diameter also would require a higher number of representatives distributed across the network.
 For a network in which there exist natural communities, we could have representatives from each community and the number of representatives from a community would depend on the size of the community. 
 In general, a good representative set would consist of nodes distributed over the network, so that they represent different sections of the network. 
 Hence it would be interesting to study how the network structure influences the ease of determining representative set and its effectiveness.

 Network compression techniques could be considered for solving the studied problem;
 the preferences with respect to a number of topics over their sets of alternatives, could be viewed as constituting the state of a node.
 Some of the relevant techniques that could be useful are 
 efficient data representations for large high-dimensional data \cite{belkin2003laplacian} using spectral graph theory and graph Laplacian \cite{chung1997spectral},
 multiple transforms for data indexed by graphs
 \cite{sandryhaila2013discrete}, etc.
 Furthermore, with the emergence of online social networks, it is possible to obtain detailed attributes and interests of a node based on its public profile, liked pages, followed events and personalities, shared posts, etc. 
 Network compression techniques could be used for arriving at a concise  set of attributes and nodes, for representing network data.

 We believe the expected weak insensitivity property introduced in this paper, could be of interest to the social choice theory community and has a scope of further study.
 It will also be interesting to study how one should weigh the preferences of the nodes in a representative set, so that the error is a monotone decreasing function of the set.
 One could consider the scenario when nodes are strategic while reporting their preferences.
 Alternative models for the spread of preferences in a network, given the edge similarities, could be studied.
  Considering the attributes of nodes and alternatives~\cite{grum2013uai} in addition to the underlying social network for determining the best representatives, is another direction worth exploring.
 
 \section*{Acknowledgments}
The original version of this paper is accepted for publication in IEEE Transactions on Network Science and Engineering
(DOI 10.1109/TNSE.2017.2772878).
 A previous version of this paper is published in {\em First AAAI Conference on Human Computation and Crowdsourcing} \cite{dhamal2013scalable}.
 This research was in part supported by an unrestricted research grant from Adobe Labs, Bangalore. We thank Balaji Srinivasan and his team for helpful discussions.
 The first author was supported by IBM Doctoral Fellowship when most of this work was done.  
 Thanks to Akanksha Meghlan, Nilam Tathawadekar, Cressida Hamlet, Marilyn George, Aiswarya S., and Chandana Dasari, Mani Doraisamy, Tharun Niranjan, and Srinidhi Karthik B. S., for helping us with the Facebook app. 
 Thanks to Prabuchandran K. J. for useful discussions.
 Many thanks to Google India, Bangalore (in particular, Ashwani Sharma) for providing us free credits for hosting our app on Google App Engine.
 Also thanks to several of our colleagues for their useful feedback on the app.
 We thank the anonymous reviewers for their useful and insightful comments, which led to improvements across all sections of the paper.

\bibliographystyle{IEEEtran}
\bibliography{PASN_journal_references}

\begin{appendices}

\section{Description of the Facebook App}
\label{app:facebook_app}

\subsection{Overview}
\label{app:overview}

Online social networking sites such as Facebook, Twitter, and Google+ are highly popular in the current age; for instance, Facebook has over 2 billion monthly active users as of 2017. Using such online social networking sites for data collection has become a trend in several research domains. 
When given permission by a user, it is easy to obtain access to the user's friend list, birthday, public profile, and other relevant information using Facebook APIs.
Furthermore, Facebook provides a facility to its users to invite their friends to use any particular application, and hence propagate it. 
Owing to these reasons, in order to obtain the data for our purpose, we developed a Facebook application titled {\em The Perfect Representer} for eliciting the preferences of users over a set of alternatives for a wide range of topics, as well as to obtain the underlying social network. 
Once a user logged into the app, the welcome page as shown in Figure~\ref{fig:intro_page} was presented, which described to the user what was to be expected from the app.

First, the user would have to give his/her preferences over 5 alternatives for 8 topics, using a drag`n'drop interface as shown in Figure~\ref{fig:questions_page}.
The user was given the option of skipping any particular topic if he/she wished to.
The topics, which were broadly classified into personal and social types, and their alternatives are listed in Table~\ref{tab:questions} (the ordering of alternatives from top to bottom is based on the aggregate preference computed from our data as per the Borda count aggregation rule).
From a user's viewpoint, the app gave the user a {\em social centrality score} out of 10, telling how well the user represents the society or how well the user's opinions are aligned with that of the society with respect to the provided preferences.
The score was dynamic and kept on updating as more users used the app (since the aggregate preference itself kept on updating); this score could be posted on the user's timeline.
The user also had an option of viewing how similar his/her preferences were to his/her selected friends. 
Explicit incentives were provided for users to propagate the app either by inviting their friends or sharing on their timelines as well as 
messaging through emails and popular websites
(Figure~\ref{fig:sharing_page}). 

\begin{figure}[t!]
\centering
\includegraphics[scale=.25]{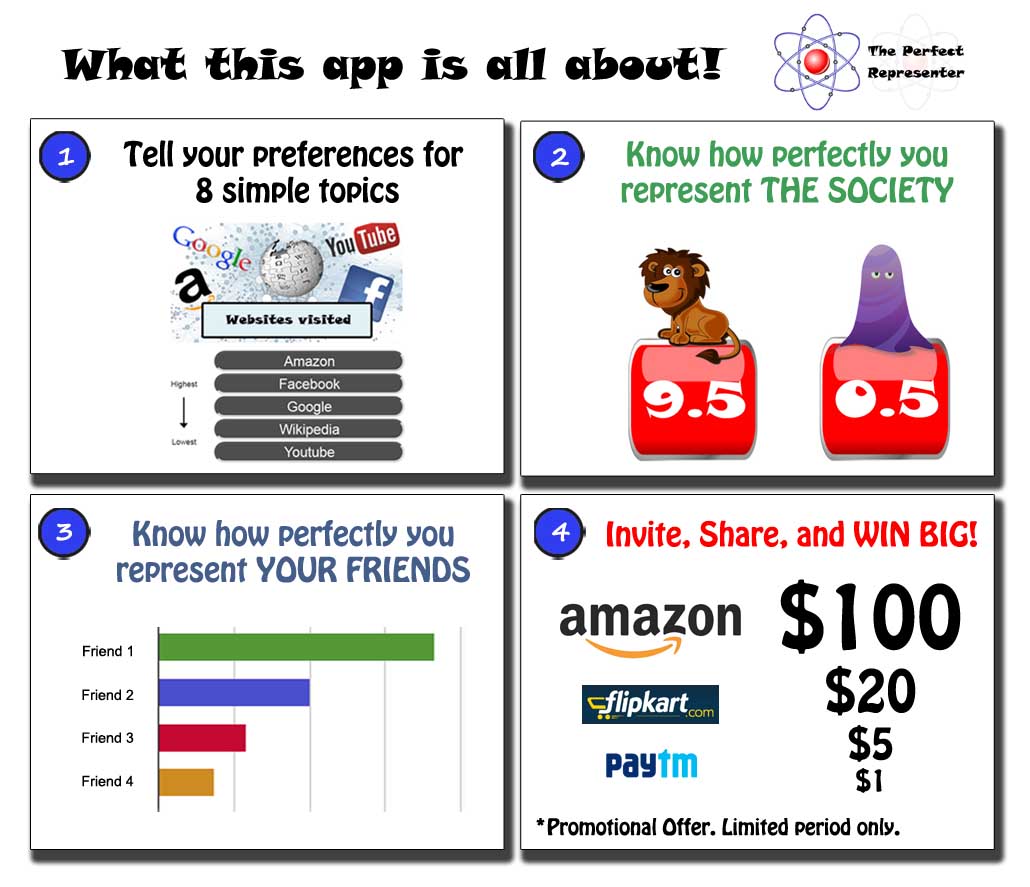}
\caption{Screenshot of the welcome page}
\label{fig:intro_page}
\vspace{4mm}
\end{figure}

\begin{figure}[t!]
\centering
\includegraphics[scale=.53]{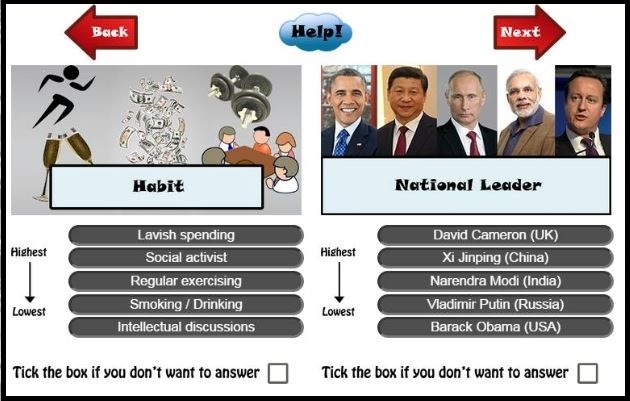}
\caption{Screenshot of a page with topics and their alternatives}
\label{fig:questions_page}
\end{figure}

To host our application, we used Google App Engine, which provides a Cloud platform for facilitating large number of hits at any given time as well as large and fast storage.

\begin{table*}
\centering
\caption{Topics and alternatives in the Facebook app}
\begin{footnotesize}
\begin{tabular}{|c|c|c|c||c|c|c|c|}
\hline
\multicolumn{4}{|c||}{\T \B Personal}	&	\multicolumn{4}{c|}{\T \B Social}	
\\ \hline \T \B
Hangout	&	Chatting 			&	Facebook 	&	Lifestyle	&	Website		&	Government 		&	Serious 	&	Leader
\\ \T \B
Place			&	App					& Activity		&						&	visited		&	Investment			& Crime	&
\\ \hline \T \B
Friend's place 						&	WhatsApp			&	Viewing Posts				 & 		Intellectual		&	Google		&	Education					&	Rape		&		N. Modi (India) 
\\ \T \B
Adventure Park	&	Facebook	&	Chatting 					 & Exercising  &	Facebook	&		Agriculture		&	Terrorism 			& B. Obama (USA)
\\ \T \B
Trekking				&	Hangouts 					&	Posting 	& Social activist						&	Youtube		&	Infrastructure 				&	Murder				&	D. Cameron (UK) 
\\ \T \B
Mall	&	SMS		&	Games/Apps		& Lavish					&	Wikipedia	&	Military						&	Corruption	&	V. Putin (Russia)
\\ \T \B
Historical Place		&	Skype					&	Marketing			& Smoking			&	Amazon	&	Space explore	&	Extortion		&	X. Jinping (China)
\\ \hline
\end{tabular}
\end{footnotesize}
\label{tab:questions}
\end{table*}

\subsection{The Scores}
\label{sec:scores}

Let $\mathbb{A}$ be the set of alternatives and $r=|\mathbb{A}|$.
Let $\tilde{c}(p,q)$ be the similarity between preferences $p$ and $q$. In our app, we implement $\tilde{c}(p,q)$ to be the normalized Footrule similarity
as a computationally efficient approximation \cite{diaconis1977spearman} to normalized Kendall-Tau similarity (which is used throughout our study) for scoring the users in real-time;
the normalized Kendall-Tau similarities are computed offline.
Let $w_a^p$ denote the position of alternative $a$ in preference $p$. The Footrule distance between preferences $p$ and $q$ is given by 
$ 
\sum_{a \in \mathbb{A}} |w_a^p - w_a^q|
$. 
With $r$ being the number of alternatives, it can be shown that the maximum possible Footrule distance is
$ 
2 \lceil \frac{r}{2}\rceil \lfloor \frac{r}{2} \rfloor
$. 
So the normalized Footrule distance between preferences $p$ and $q$ can be given by
\begin{displaymath}
\tilde{d}(p,q) = 
\frac{\sum_{a \in \mathbb{A}} |w_a^p - w_a^q|}{2 \lceil \frac{r}{2}\rceil \lfloor \frac{r}{2} \rfloor}
\end{displaymath}
and normalized Footrule similarity by $\tilde{c}(p,q) = 1-\tilde{d}(p,q)$.

For example, the normalized Footrule similarity between preferences $p=(A,B,C,D,E)$ and $q=(B,E,C,A,D)$ is $\tilde{c}(p,q)=\left( 1-\frac{|1-4|+|2-1|+|3-3|+|4-5|+|5-2|}{2 \lceil \frac{5}{2}\rceil \lfloor \frac{5}{2} \rfloor} \right) = \frac{1}{3}$.

\begin{figure}[t!]
\centering
\includegraphics[scale=.195]{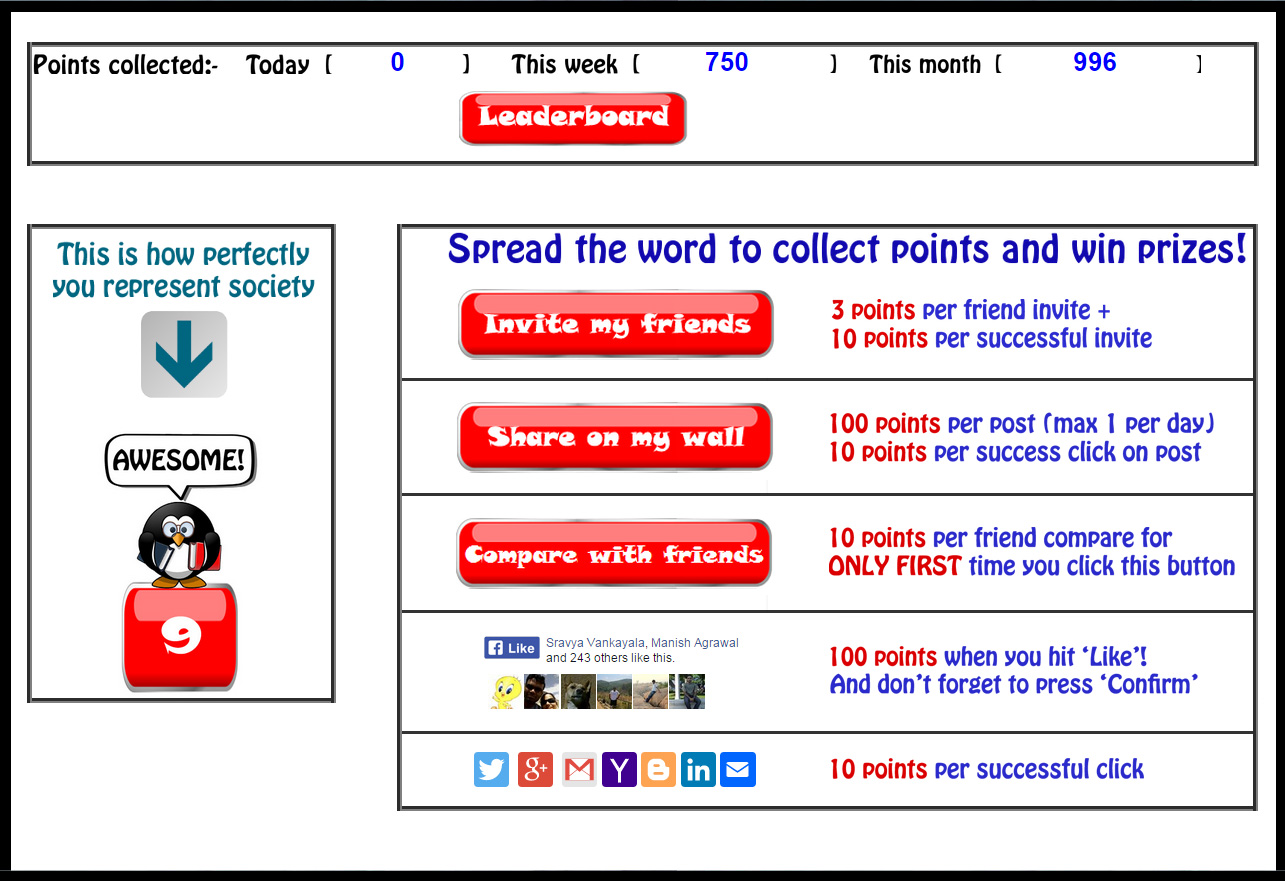}
\caption{Screenshot of the sharing and scoring page}
\label{fig:sharing_page}
\vspace{-1mm}
\end{figure}

\vspace{2mm}
\subsubsection{Social Centrality Score - How Perfectly you Represent the Society?}

Let $p_{it}$ be the preference of node $i$ for topic $t$, and $p_{At}$ be the aggregate preference of the population for topic $t$. For the purpose of our app's implementation, we obtain the aggregate preference using the Borda count rule.
For computing the social centrality, we give each topic $t$, a weight proportional to $n_t$ (the number of users who have given their responses for that topic).
So the fractional score of user $i$ is given by
$ 
\sum_t \left( \frac{n_t}{\sum_t n_t} \right) \tilde{c}(p_{it} , p_{At})
$.

A primary drawback of standard measures such as Kendall-Tau or Footrule distance
 is the way distances between preferences themselves are distributed. For instance, there are several preferences which are at a distance of, say 0.5, from a given preference as compared to those at a distance of, say 0.3, which leads to a bias in distances to be concentrated in the intermediate range. 
Since the distance between two preferences (here, user preference and the aggregate preference) for most topics would be concentrated in the intermediate range, a user would seldom get a very high or a very low score. A mediocre score would, in some sense, act as a hurdle in the way of user sharing the post on his/her timeline. So to promote posting their scores, we used a simple boosting rule (square root) and then enhanced it to the nearest multiple of 0.5, resulting in the final score of 
\begin{displaymath}
\frac{1}{2}\;  \left\lceil 20 \;\sqrt{\sum_t \left( \frac{n_t}{\sum_t n_t} \right) \tilde{c}(p_{it} , p_{At})} \;\right\rceil \;\;\;\text{(out of 10).}
\end{displaymath}

\subsubsection{How Perfectly you Represent your Friends?}

Once a user selected a list of friends to see how similar they are to the user, the app would give the similarity for each friend in terms of percentage. This similarity was also a function of the number of common questions they responded to. So the similarity between nodes $i$ and $j$ in terms of percentage was given by
\begin{displaymath}
100 \left( \frac{\sum_t \tilde{c}(p_{it} , p_{jt})}{\sum_t 1} \right)
\end{displaymath}
where $\tilde{c}(p_{it} , p_{jt})=0$ if either $i$ or $j$ or both did not respond to topic $t$.

\subsection{The Incentive Scheme}
\label{sec:incentives}

A typical active Facebook user uses several apps in a given span of time, and invites his/her friends to use it depending on the nature of the app and the benefits involved. 
In order to ensure a larger reach, it was important to highlight the benefits of propagating our app. We achieved this by designing a simple yet effective incentive scheme for encouraging users to propagate the app by sharing it on their timelines and inviting their friends to use it. 

We incorporated a points system in our app, where suitable points were awarded on a daily, weekly, as well as on an overall basis, for spreading the word about the app through shares, invites, likes, etc. Bonus points were awarded when a referred friend used the app through the link shared by the user. 
 To ensure competitiveness in sharing and inviting, the daily and weekly `top 10' point collectors were updated in real-time and the winners were declared at 12 noon GMT (daily) and Mondays 12 noon GMT (weekly). 
A cutoff was set on the number of points to be eligible to get a prize. 
Users were also given a chance to win a big prize through daily, weekly, and bumper lucky draws, if they crossed a certain amount of points, so that users with less number of friends could also put their effort even though they did not have a chance to make it into the `top 10'.
Prizes were awarded in the form of gift coupons so that getting the prize was in itself, quick as well as hassle-free for the users.

The points structure as well as the links to invite, like, and share were provided on the scoring page (Figure~\ref{fig:sharing_page}), giving the users a clear picture of how to earn points and win prizes.
The lists of daily and weekly winners were displayed on the welcome page and the scoring page.

\begin{figure}[t!]
\hspace{-10mm}
\includegraphics[scale=.3]{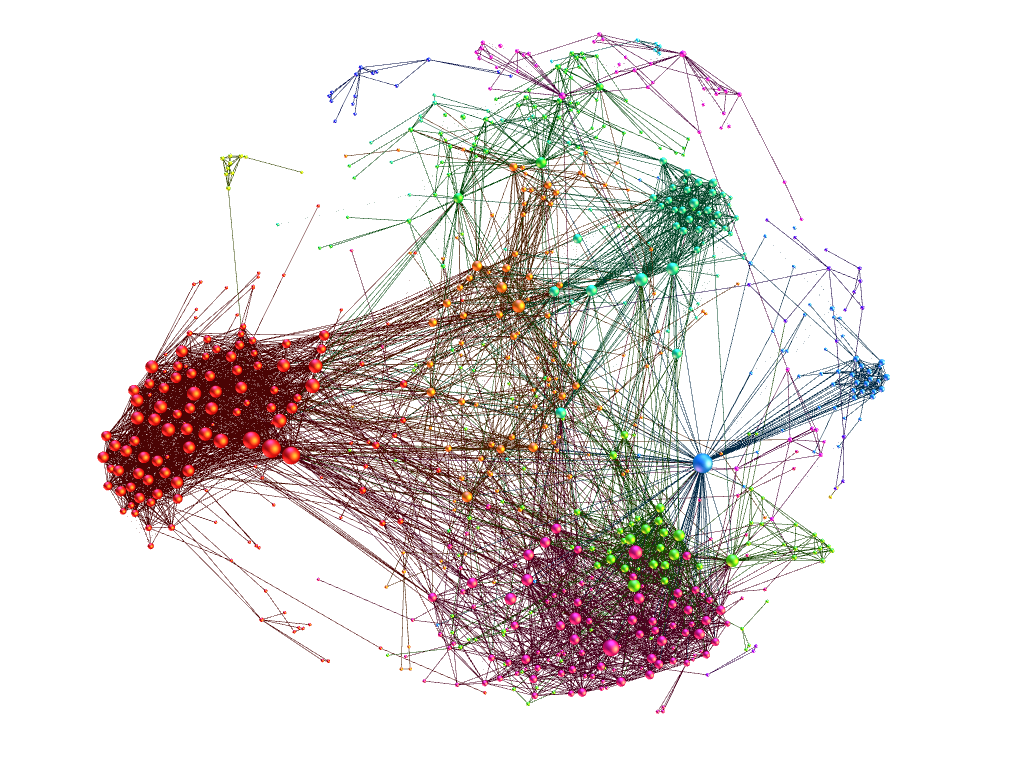}
\vspace{-7mm} 
\caption{Network of the app users}
\label{fig:network_app_users}
\end{figure}

 \begin{table}[t]
 \centering
 \caption{App user statistics}
 \small
 \begin{tabular}{|c|c|c|}
 \hline
 \T \B
         Attribute & Value & Count
 \\ \hline \hline
 \multirow{9}{*}{Age}
 &	13-18 &	25
 \\
 &	19-21 &	120
 \\
 &	22-24 &	291
 \\
 &	25-28 &	272
 \\
 &	29-32 &	58
 \\
 &	33-40 &	37
 \\
 &	41-50 &	14
 \\
 &	51-70 &	12
 \\
 &	Indeterminate &	15
 \\ \hline
 \multirow{7}{*}{Location}
 &	North America	&	44
 \\
 &	Europe	&	18
 \\
 &	Middle east	&	12
 \\
 &	India	&	749
 \\
 &	East Asia	&	8
 \\
 &	Australia	&	5
 \\
 &	Indeterminate	&	8
 \\ \hline
 \multirow{3}{*}{Gender}
 & Male & 565
 \\
 & Female & 269
 \\
 & Indeterminate & 10
 \\ \hline
 \end{tabular}
 \label{tab:user_stats}
 \end{table}

\subsection{Preprocessing of the Obtained Dataset}

The obtained dataset consisted of 1006 nodes and 7112 edges.
Figure~\ref{fig:network_app_users} shows the network of app users and Table \ref{tab:user_stats} shows some user statistics.
It was necessary to preprocess the data before using it for model-fitting. For instance, we eliminated nodes which provided responses to fewer than 6 topics so as to consider only those nodes which responded to sufficient number of topics.
Further, in order to observe the network effect, it was necessary that all nodes belonged to the same component; so we considered only the giant component which consisted of 844 nodes and 6129 edges.

The obtained network is a subgraph of nodes who have used the Facebook app to give their preferences for the asked topics. 
The users had the option of inviting their friends to use the app and sharing their app social centrality scores on their timelines, and were provided incentives for the same.
It is well known that users generally use an app suggested by their close friends and also invite their close friends to use an app. 
Similarly, the social centrality posts shared by users are more likely to be viewed and clicked on by their close friends.
We observed that connected nodes had mean dissimilarity ($=\avg_{(i,j)\in E} d(i,j)$) of 0.24
and unconnected nodes had mean dissimilarity of 0.37.
Due to the way that the app usage is propagated, it is likely for most nodes to have their close friends included in the network. This is an intuition behind the observed homophily in the obtained network.

\section{Proof of Proposition~\ref{prop:nphard}}
\label{app:nphard_proof}
\begin{customprop}{\ref{prop:nphard}}
Given constants $\chi$ and $\omega$, \\
(a) it is NP-hard to determine whether there exists a set $M$ consisting of $k$ nodes such that $\rho(M) \geq \chi$, and \\
(b) it is NP-hard to determine whether there exists a set $M$ consisting of $k$ nodes such that $\psi(M) \geq \omega$.
\end{customprop}
\begin{proof}
We reduce an NP-hard Dominating Set problem instance to the problem under consideration.
Given a graph $G$ of $n$ vertices, the dominating set problem is to determine whether there exists a set $D$ of $k$ vertices such that every vertex not in $D$, is adjacent to at least one vertex in $D$.

Given a dominating set problem instance, we can construct a weighted undirected complete graph $H$ consisting of the same set of vertices as $G$ such that, the weight $c(i,j)$ of an edge $(i,j)$ in $H$ is some high value (say $0.9$) if there is edge $(i,j)$ in $G$, else it is some low value (say $0.6$). 

Now there exists a set $D$ of $k$ vertices in $G$ such that the distance between any vertex in $G$ and any vertex in $D$ is at most one, if and only if there exists a set $M$ of $k$ vertices in $H$ such that $\rho(M) \geq 0.9$ or $\psi(M) \geq k+0.9(n-k)$. Here $\chi=0.9$ and $\omega=k+0.9(n-k)$.
This shows that the NP-hard dominating set problem is a special case of the problems under consideration, hence the result.
\end{proof}

\section{Proof of Proposition~\ref{prop:submodularity}}
\label{app:submod_proof}
\begin{customprop}{\ref{prop:submodularity}}
The objective functions $\rho(\cdot)$ and $\psi(\cdot)$ 
are non-negative, monotone, and submodular.
\end{customprop}
\begin{proof}
We prove the properties in detail for $\psi(\cdot)$. The proof for $\rho(\cdot)$ is similar.

Consider sets $S,T$ such that $S \subset T \subset N$ and a node $v \in N \setminus T$. It is clear 
that $\psi(\cdot)$ is non-negative.
Let 
$x_i = c(S,i),\;y_i = c(T,i),\;\bar{x_i} = c(S \cup \{v\} ,i),\;\bar{y_i} = c(T \cup \{v\} ,i)$.
For any $i \in N$,
\begin{displaymath}
c(S,i) = \max_{j \in S \subseteq T} c(j,i) \leq \max_{j \in T} c(j,i) = c(T,i)
\end{displaymath}
\begin{equation}
\label{eqn:x_y}
\implies x_i \leq y_i
\end{equation}
That is, $\psi(\cdot)$ is monotone. Similarly, it can be shown that
\begin{equation}
\label{eqn:three}
\bar{x_i} \leq \bar{y_i} \; ; \;\;\;
x_i \leq \bar{x_i} \; ; \;\;\;
y_i \leq \bar{y_i}
\end{equation}
\begin{eqnarray}
\text{Now,} \; \;\;
y_i < \bar{y_i} 
\label{eqn:maxT} & \implies & k \notin \argmax_{j \in T \cup \{v\}} c(j,i) \text{ } \forall k \in T \text{ }\\
\label{eqn:maxS} & \implies & k \notin \argmax_{j \in S \cup \{v\}} c(j,i) \text{ } \forall k \in S \subseteq T \text{ }\;\;\;\;\;\\
\label{eqn:y_ybar_x_xbar}
& \implies & x_i < \bar{x_i} 
\end{eqnarray}
The contrapositive of the above, from Inequalities~(\ref{eqn:three}) is
\begin{equation}
\label{eqn:contra}
x_i = \bar{x_i} \implies y_i = \bar{y_i} 
\end{equation}
Also, from Implications~(\ref{eqn:maxT}) and (\ref{eqn:maxS}),
\begin{eqnarray}
\nonumber 
y_i < \bar{y_i} 
& \implies & \{v\} = \argmax_{j \in T \cup \{v\}} c(j,i) = \argmax_{j \in S \cup \{v\}} c(j,i) \\
\label{eqn:xbar_equals_ybar} & \implies & \bar{x_i} = \bar{y_i}
\end{eqnarray}
Now from Inequalities~(\ref{eqn:three}), depending on node $i$, four cases arise that relate the values of $\bar{x_i} - x_i$ and $\bar{y_i} - y_i$.
\begin{description}
\item[Case 1:] ~$x_i = \bar{x_i}$ and $y_i = \bar{y_i}$:\\
In case of such an $i$, we have $\bar{x_i} - x_i = \bar{y_i} - y_i$
\item[Case 2:] ~$x_i = \bar{x_i}$ and $y_i < \bar{y_i}$:\\
By Implication~(\ref{eqn:contra}), there does not exist such an $i$.
\item[Case 3:] ~$x_i < \bar{x_i}$ and $y_i = \bar{y_i}$:\\
In case of such an $i$, we have $\bar{x_i} - x_i > \bar{y_i} - y_i$
\item[Case 4:] ~$x_i < \bar{x_i}$ and $y_i < \bar{y_i}$:
For such an $i$,
\begin{eqnarray}
\hspace{-4mm}
\nonumber
\bar{x_i} - x_i &=& \bar{y_i} - x_i \;\;\;\text{ (from~(\ref{eqn:y_ybar_x_xbar}) and (\ref{eqn:xbar_equals_ybar}))} \\
\nonumber
&\geq& \bar{y_i} - y_i \;\;\;\text{ (from Inequality~(\ref{eqn:x_y}))} 
\end{eqnarray}
\end{description} 

From the above cases, we have 
\begin{eqnarray}
\hspace{-4mm}
\nonumber
& & \bar{x_i} - x_i \geq \bar{y_i} - y_i, \;\;\;\forall i \in N \\
\nonumber
&\implies& \sum_{i \in N} (\bar{x_i} - x_i) \geq \sum_{i \in N} (\bar{y_i} - y_i) \\
\nonumber
&\implies& \sum_{i \in N} \bar{x_i} - \sum_{i \in N} x_i 
\geq \sum_{i \in N} \bar{y_i} - \sum_{i \in N} y_i \\
\nonumber
&\implies& \psi(S \cup \{v\}) - \psi(S) \geq \psi(T \cup \{v\}) - \psi(T)
\end{eqnarray}

As the proof is valid for any $v \in N \setminus T$ and for any $S,T$ such that $S \subset T \subset N$, the result is proved.
\end{proof}

\section{Proof of Theorem~\ref{shapeqnu}}
\label{app:shapeqnu}
\begin{customthm}{\ref{shapeqnu}}
 For the TU game defined by $\nu(S) = \sum_{\substack{i,j \in S\\ i \neq j}} c(i,j) $,
$ 
\nonumber
 \phi(\nu)=Nu(\nu)=Gv(\nu)=\tau(\nu).
$ 
\end{customthm}
\begin{proof}
The characteristic function $\nu(S) = \sum_{\substack{i,j \in S\\ i \neq j}} c(i,j) $,
when $|S|=2$ where $S=\{i,j\}$, 
\begin{equation}
\label{eqn:for2}
\nu(\{i,j\})=c(i,j)
\end{equation}
\begin{equation}
\label{eqn:pair}
\therefore\; \nu(S) = \sum_{\substack{T \subseteq N\\ |T|=2}} \nu(T)
\end{equation}
From Equation~(\ref{eqn:for2}),
 the Shapley value 
 can be rewritten as
\begin{align}
\label{shapleyeqhalf}
\hspace{-2mm}
\phi_j(\nu) = \frac{1}{2} \sum_{\substack{i \in N\\i \neq j}}c(i,j)
= \frac{1}{2} \sum_{\substack{i \in N\\i \neq j}}\nu(\{i,j\})
= \frac{1}{2}\sum_{\substack{S \subseteq N\\ j \in S\\|S|=2}} \nu(S)
\end{align}
~~
The proof for $\phi(\nu)=Nu(\nu)$ follows from \cite{chun2007coincidence}.
  Furthermore, it has been shown in \cite{chun2007coincidence} that, for the TU game satisfying Equation~(\ref{shapleyeqhalf}), for each $S \subseteq N$,
 \begin{equation}
  \nonumber
  \nu(S) - \sum_{i \in S}\phi_i(\nu) = \nu(N\backslash S) - \sum_{i \in N\backslash S}\phi_i(\nu)
 \end{equation}
 \begin{equation}
 \nonumber
 \therefore\;
  \nu(\{i\}) - \phi_i(\nu) = \nu(N\backslash\{i\}) - \sum_{\substack{j\in N\\j\neq i}}\phi_j(\nu)
    \text{ , for $S = \{i\}$}
 \end{equation}
 So, the propensity to disrupt for player $i$ \cite{straffin1993game} for the Shapley value allocation is
 \begin{equation}
 \nonumber
 d_i(\phi(\nu)) = \frac{\sum_{j\in N,j\neq i}\phi_j(\nu) - \nu(N\backslash \{i\})}{\phi_i(\nu) - \nu(\{i\})} = 1
 \end{equation}
 As the propensity to disrupt is equal for all the players ($=1$), 
 this allocation is the Gately point \cite{straffin1993game}, that is,
 $
  \phi(\nu)=Gv(\nu).
  $

Let $M(\nu)=(M_i(\nu))_{i\in N}$ and $m(\nu)=(m_i(\nu))_{i\in N}$, where $M_i(\nu)=\nu(N)-\nu(N\backslash \{i\})$ and $m_i(\nu)=\nu(\{i\})$.
For a convex game, 
$
\tau(\nu)= \lambda M(\nu) + (1-\lambda)m(\nu)
$,
where $\lambda\in [0,1]$ is chosen such that \cite{chun2007coincidence},
\begin{equation}
\label{tausatis}
\sum_{i \in N} \left[ \lambda M_i(\nu) + (1-\lambda) m_i(\nu) \right] = \nu(N)
\vspace{-1mm}
\end{equation}
\begin{equation*}
\text{From (\ref{eqn:pair}), }
M_i(\nu)
= \sum_{\substack{S\subseteq N \\ |S| = 2}}\nu(S) - \sum_{\substack{S\subseteq N\backslash \{i\} \\ |S| = 2}}\nu(S) 
= \sum_{\substack{S\subseteq N \\ i \in S \\ |S| = 2}}\nu(S)
\vspace{-2mm}
\end{equation*}
This, with Equation~(\ref{tausatis}) and the fact that for our game, for all $i$, $m_i(\nu) = \nu(\{i\}) = 0$,
\begin{equation*}
 \nu(N) \;=\; \lambda\sum_{i\in N}M_i(\nu) 
\;=\; \lambda\sum_{i\in N} \sum_{\substack{S\subseteq N \\ i \in S \\ |S| = 2}}\nu(S)
\;=\; 2\lambda\sum_{\substack{S\subseteq N \\ |S| = 2}}\nu(S) \\ 
\end{equation*}
Using Equation~(\ref{eqn:pair}), we get $\lambda = \frac{1}{2}$.
So we have
\begin{equation}
 \nonumber
\tau_i(\nu) = \frac{1}{2}M_i(\nu) + \frac{1}{2}m_i(\nu) = \frac{1}{2}\sum_{\substack{S \subseteq N\\ i \in S\\|S| = 2}} \nu(S)
\end{equation}
This, with 
Equation~(\ref{shapleyeqhalf}), gives
$
\phi(\nu) = \tau(\nu)
$
\end{proof}

\end{appendices}

\end{document}